\documentclass{amsart}
\usepackage{mathrsfs}
\usepackage[abbrev]{amsrefs}
\usepackage{amssymb}
\usepackage{graphicx}
\usepackage{color}
\usepackage{float}
\usepackage{hyperref}
\usepackage{tikz,pgflibraryshapes}
\usepackage{enumitem}
\usepackage[margin=1.37in]{geometry}
\usetikzlibrary{arrows}
\usepackage{bbm}

\usepackage{CJK}

\usepackage{float}

\usepackage{changes}
\usepackage{soul}
\setstcolor{red}

\definecolor{green}{rgb}{0.0, 0.5, 0.5}
\definecolor{yellow}{rgb}{0.5, 0.5, 0}
\definecolor{lgray}{gray}{0.9}
\definecolor{llgray}{gray}{0.95}
\definecolor{lllgray}{gray}{0.975}

\theoremstyle{plain}%
\newtheorem{theorem}{Theorem}
\newtheorem{proposition}[theorem]{Proposition}%
\newtheorem{lemma}[theorem]{Lemma}

\theoremstyle{remark}%
\newtheorem{remark}{Remark}

\theoremstyle{definition}%
\numberwithin{equation}{section}
\numberwithin{theorem}{section}
\numberwithin{corollary}{section}

\newcommand{\inn}[2]{\left\langle#1,\,#2\right\rangle}

\DeclareMathOperator{\supp}{supp}

\newcommand{\di}{\partial}

\newcommand{\br}[1]{\left\langle#1\right\rangle}
\newcommand{\roun}[1]{\left(#1\right)}

\newcommand{\eps}{\epsilon}

\newcommand{\g}{\gamma}
\newcommand{\al}{\alpha}

\newcommand{\Rb}{\mathbb{R}}

\newcommand{\om}{\omega}

\renewcommand{\l}{\lambda} 

\newcommand{\abs}[1]{\ensuremath{\left\lvert#1\right\rvert}}

\newcommand{\sbr}[1]{\left[#1\right]}
\newcommand{\Set}[1]{\left\{#1\right\}}

\newcommand{\md}[6]{\ensuremath{
		\ifinner
		\tfrac{\partial{^{#2}}#1}{\partial{#3^{#4}}\partial{#5^{#6}}}
		\else
		\tfrac{\partial{^{#2}}#1}{\partial{#3^{#4}}\partial{#5^{#6}}}
		\fi
}}
\newcommand{\del}[1]{\left(#1\right)}
\newcommand{\thmref}[1]{Theorem~\ref{#1}}

\newcommand{\secref}[1]{Section~\ref{#1}}
\newcommand{\lemref}[1]{Lemma~\ref{#1}}
\newcommand{\propref}[1]{Proposition~\ref{#1}}

\newcommand{\figref}[1]{Figure~\ref{#1}}

\newcommand{\cD}{\mathcal{D}}

\newcommand{\cE}{\mathcal{E}}
\newcommand{\cF}{\mathcal{F}}

\newcommand{\cL}{\mathcal{L}}
 
\newcommand{\nc}{\newcommand}

\nc{\h}{\delta}
\nc{\G}{\Gamma}
\nc{\et}{\eta} 
\nc{\gam}{\gamma}
\nc{\ka}{\kappa}
\nc{\lam}{\lambda}
\nc{\Lam}{\Lambda}

\nc{\ta}{\tau}
\nc{\w}{\omega}
\nc{\io}{\iota}
\nc{\s}{\sigma}
\nc{\vphi}{\varphi}

\nc{\e}{\epsilon}

\nc{\ran}{\rangle}
\nc{\lan}{\langle}

\nc{\bfone}{{\bf 1}}

\newcommand{\p}{\partial}

\nc{\dd}{\mathrm{d}}

\newcommand{\DETAILS}[1]{}

\DeclareMathOperator{\Rem}{Rem}
\DeclareMathOperator{\dG}{\mathrm{d}\Gamma}

\newcommand{\ondel}[1]{\inn{\varphi}{{#1}\varphi}}

\newcommand{\wzdel}[1]{\inn{\psi_0}{{#1}\psi_0}}

\newcommand{\wtdel}[1]{\inn{\psi_t}{{#1}\psi_t}}

\newcommand{\Mr}[1]{\mathrm{#1}}

\nc{\hc}{\mathrm{h.c.}}
\nc{\TRem}{\widetilde \Rem}
\nc{\HRem}{\widehat \Rem}

\begin{document}
	
	\title[Microscopic propagation of long-range quantum many-body systems]{On the microscopic propagation speed of long-range quantum many-body systems}
	
	\author{Marius Lemm}
	\address{Department of Mathematics, University of T\"ubingen, 72076 T\"ubingen, Germany }
	\email{marius.lemm@uni-tuebingen.de}
	
	\author{Carla Rubiliani}
	\address{Department of Mathematics, University of T\"ubingen, 72076 T\"ubingen, Germany }
	\email{carla.rubiliani@uni-tuebingen.de}
	
	\author{Jingxuan Zhang
	}
	\address{Yau Mathematical Sciences Center, Tsinghua University, Beijing 100084, China }
	\email{jingxuan@tsinghua.edu.cn}
	
	\date{November 14, 2023}

	\begin{abstract}
		
		We consider the time-dependent Schr\"odinger equation that is generated on the bosonic Fock space by a long-range quantum many-body Hamiltonian. We derive the first bound on the maximal speed of particle transport in these systems that is thermodynamically stable and holds all the way down to microscopic length scales. For this, we develop a novel multiscale rendition of the ASTLO (adiabatic spacetime localization observables) method. Our result opens the door to deriving the first thermodynamically stable Lieb-Robinson bounds on general local operators for these long-range interacting bosonic systems.
		
	\end{abstract}
	\maketitle
	
	\section{Introduction}
	
	Let $\Lam$ be a finite subset of a lattice $\cL\subset \Rb^d,\,d\ge1$.
	We consider {solutions of the} the many-body Schr\"odinger equation
	\begin{align}\label{SE}
		i\di_t\psi_t = H_\Lam \psi_t\qquad \textnormal{ on } \cF(\ell^2(\Lam)),
	\end{align}
 where $\cF(\ell^2(\Lam))$ denotes the bosonic Fock space and $H_\Lam$ is a long-range quantum many-body  Hamiltonian {acting} on $\cF(\ell^2(\Lam))$; see \eqref{H}--\eqref{eq:fock} for the definition.

In this paper, we rigorously identify simple sufficient conditions on the Hamiltonian and the initial states to obtain a bound on the maximal particle propagation speed of solutions to \eqref{SE}. Specifically, 
	we prove that for $H_\Lam$ with polynomial decaying hopping and arbitrary density-density interaction terms, 
	any solution $\psi_t$ to \eqref{SE} with initial density uniformly bounded  by $\l>0$ (i.e., $\sup_{x\in\Lam}\wzdel{n_x}\le \l$) obeys the following estimate: There exist $\kappa,\,C>0$ independent of $\l,\,\Lam$, and $\psi_0$, such that for any $v>\kappa,\, r>0$, and $\eta>1$,
	\begin{align}
		\label{MVB0}
		\sup _{0\le {t}<\eta/v}\wtdel{N_{B_r}} \le (1+C\eta^{-1})\inn{\psi_0}{{N_{B_{r+\eta}}}\psi_0}+C\l.
	\end{align}
	Here $B_h:=\Set{x\in\Lam:\abs{x}\le h}$ and $N_X$ denotes the particle number operator for $X\subset \Lam$. 
	We also obtain analogous estimates to \eqref{MVB0} for arbitrary moments $\wtdel{N_{B_r}^p}$ for $p\ge1$ {under stronger (but still polynomial) decay assumptions on the hopping}. Higher moment bounds are useful in many applications because they provide stronger tail bounds via Markov's inequality. See \secref{secIntro} for detailed setup and \thmref{pMVB main} for  precise statement of our main result for all $p\geq 1$.
	
	In physical terms, our propagation estimate \eqref{MVB0} ensures that bosons with long-range hopping and interactions {move at most with speed $\kappa$} as long as the initial state has uniformly bounded density, see \eqref{UDBp}. (Bounded-density initial states include the physically most relevant class of \textit{Mott states}, which are tensor product states in the particle number eigenbasis with uniformly bounded local eigenvalues, that are also studied experimentally \cite{cheneau2012light,cheneau2022experimental}.) 
 In fact, our method yields an explicit bound on the maximal speed of particle propagation $\kappa$, see \eqref{kappa}, which has a natural interpretation as a one-particle group velocity operator. In particular, $\kappa$ is independent of time and total particle number. 

 {Controlling bosonic propagation is a central topic of mathematical physics. The main challenge is that bosons can in principle accumulate in unbounded ways, so operator norms need to be avoided. Bosons with long-range hopping and interactions combine the obstacles coming from (a) unboundedness and (b) long-range nature of the terms in the Hamiltonian. Here we develop a method to overcome these obstacles, a new multiscale rendition of the ASTLO (adiabatic spacetime localization observables) method, which combines recursive differential inequalities and microlocal-inspired commutator expansions. We expect the approach to be useful more broadly for analyzing quantum dynamics problems which combine features (a) and (b).}

	\subsection{Organization}
	In Section \ref{secIntro}, we lay out our setup in detail and present the main result, \thmref{pMVB main}. Afterwards, we compare our result to the literature (Section \ref{sec:rr}) and summarize the proof strategy and key challenges (Section \ref{sec:proofstrategy}).

	In \secref{secPrelim}, we set up the proof with various definitions and preliminary lemmas. In \secref{sec proof MVB}, we prove \thmref{pMVB main} for the first moment case $p=1$. In \secref{sec proof multi} we complete the proof of \thmref{pMVB main} for all higher moments.
	
	\subsection{Notation}

	Throughout the paper, we fix a lattice domain $\Lam\subset L$. 	We denote by $\cD(A)$ the domain of an operator $A$ and by
	$\|\cdot\|$   the norm of  operators on the bosonic Fock space $\cF(\ell^2 (\Lam))$. 	We make no distinction in notation between a function $f$ on $\Lam$ and the associated multiplication operator $\psi(x)\mapsto f(x)\psi(x)$.
	
	We write $\abs X= \#\Set{x\in \Lam: x\in X}$ for $X\subset \Lam$, $B_h=\Set{x\in\Lam:  \abs{x}\le h}$, and $S_h=\di B_h = \Set{x\in\Lam: \abs x= h}$ for $h>0$. 
	
		
		\section{Setup and Main Results}\label{secIntro}
		Let $\mathcal{L}\subset\mathbb{R}^d$ be a lattice equipped with the Euclidean metric and fix a finite $\Lam\subset\cL$. We assume that any two distinct sites on $\cL$ are separated by a distance greater than one and there exist constants $\omega_d,V_d>0$ such that the following polynomial growth conditions are verified for all $h>0$:
		\begin{align}
			\label{growthCond1}
			\abs{B_h}\le& V_d h^d,\\
			\abs{S_h}\le& \om_{d-1} h^{d-1}.\label{growthCond2}
		\end{align}

		We consider  quantum many-body  systems described by long-range  Hamiltonians. {The Hamiltonian is a linear unbounded operator of the forms}
		\begin{align}\label{H} 
			\boxed{H_\Lam =  \sum_{x,y \in \Lambda} J_{xy} a_x^*a_y + V}
		\end{align} 
		acting on the usual bosonic Fock space 
		\begin{equation}\label{eq:fock}
		\mathcal F(\ell^2(\Lambda))=\mathbb C\oplus  \bigoplus_{N=1}^\infty \mathcal S_N\left(\bigotimes_{j=1}^N \ell^2(\Lambda)\right),
		\end{equation}
		with $\mathcal S_N$ denoting the projection onto the permutation-symmetric subspace of $\bigoplus_{N=1}^\infty \bigotimes_{j=1}^N \ell^2(\Lambda)$.
		Here $a_x$ and $a^*_x$ are the bosonic annihilation and creation operators, respectively, {which means that they satisfy the canonical commutation relations $[a_x,a_y]=[a_x^\dagger,a_y^\dagger]=0$ and $[a_x,a_y^\dagger]=\delta_{x,y}$.} Moreover,  
  $J=(J_{xy})$ is a Hermitian $|\Lambda|\times |\Lambda|$ matrix representing the energy of individual particles, and $V=\Phi(\Set{n_x}_{x\in\Lam})$ is an arbitrary real-valued function of the bosonic number operators, $n_x:=a_x^*a_x$, which
		describes an arbitrary density-density interaction.
		{For background on the second quantization formalism and quantum many-body systems, see for example \cites{MR1441540, MR3012853, MR4470244}. In particular, \cite[App.\ A]{MR4470244} proves $H_\Lam$ is self-adjoint on the dense domain $\mathcal D(H_\Lam)=\{(\psi_N)_{N\geq 0}\in \mathcal F(\ell^2(\Lambda))\,:\, \sum_{N\geq 0}\|H_\Lambda \psi_N\|^2<\infty \}$ in $\mathcal F(\ell^2(\Lambda))$.

			\subsection{Assumptions and main result}
			Recall that $N_X:=\sum_{x\in X}n_x$ is the particle number operator for $X\subset \Lam$. We consider solutions $\psi_t$ to the many-body Schr\"odinger equation \eqref{SE} with initial datum $\psi_0$ satisfying the following uniform density bound: For some $\l>0$ and integer $p\ge1$, 
			{\begin{align}\label{UDBp}
				\wzdel{ N_X^p}\le(\l|X|)^p \text{ for all }X\subset \Lam.
			\end{align}}
%

			{Our main assumption for the Hamiltonian \eqref{H} is that} the hopping matrix $J=(J_{xy})$ satisfies the following power-law decay for some $C_J>0$ and  power $\al>2d+1$,
			\begin{align}\label{Jcond}
				\abs{J_{xy}}\le& C_J \abs{x-y}^{-\al}\qquad (x,y\in\Lam,\,x\ne y).
			\end{align}
			Since $\al>2d+1$, it follows that for the integer $n:=\lfloor\al-d-1\rfloor$, the sum $\sum_{y\in\cL}|x-y|^{n+1-\al}$ is finite for all $x\in\cL$. {This, together with {\eqref{Jcond} and the separability assumption $\abs{x-y}\ge1$ for all $x\ne y$,} implies that}
			there exist $0<\kappa_0,\ldots,\kappa_n<\infty$ depending only on $C_J$ and $\al$ such that
			\begin{align}
				\sup_{x\in\Lam}\sum_{y\in\Lam}\abs{J_{xy}}\abs{x-y}^{\nu+1}\le& \kappa_\nu\qquad (\nu=0,\ldots,n).\label{Kcond}
			\end{align}
			A central role is played by $\kappa_0$, the first moment of the hopping matrix, which as we will see bounds the many-body propagation velocity. 
			\begin{equation}\label{kappa}
				\kappa\equiv \kappa_0=\sup_{x\in\Lam}\sum_{y\in\Lam}\abs{J_{xy}}\abs{x-y}
			\end{equation}
			As we will demonstrate below, $\kappa$  yields an upper bound on the maximal velocity (i.e., the light cone slope). 
			Finally, to simplify notations, given an operator $A$ and initial state $\psi_0$, we write 
			\begin{align}
				\br{A}_t:=\br{\psi_t,A\psi_t},\quad \psi_t:=e^{-iH_\Lam t}\psi_0.
			\end{align}
			The main result of this paper is the following particle propagation bound for all moments.

			\begin{theorem}[Main result]\label{pMVB main}
				Let $p\ge1$ be an integer and assume that \eqref{Jcond} holds with 
    \begin{align}\label{alCond}
    {\al>\max\{3dp/2+1,2d+1\}}.
    \end{align}
				Then, for any ${v} > \kappa$
				{and $\delta_0>0$, there exists a positive constant ${C=C(\al,d, C_J, v,\delta_0,p)}$ such that for all $\l,\,R,\,r>0$ with $R-r>\max(\delta_0r,1)$  and initial states ${\psi_0\in\mathcal D(H_\Lam)\cap \cD(N^{p/2}_\Lam)}$ satisfying \eqref{UDBp},}
				\begin{align}\label{MVBp}
					\sup_{0\le{t}<(R-r)/v}\br{N_{B_r}^p}_t\le {\left(1+ C(R-r)^{-1}\right)} {{\br{N_{B_{R}}^p}_0}} + {C\l^p}.
				\end{align}
				
			\end{theorem}
			
			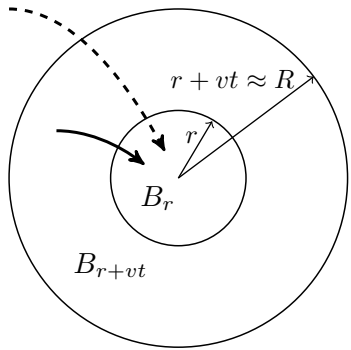
\begin{figure}[H]
				\centering
				
				\begin{tikzpicture}[scale=0.9]
					
					\draw[black,line width=0.2mm]
					(0,0) circle (1.)
					(0,0) circle (2.5);
					\draw[->,>=stealth', line width=0.4mm] (-1.8,0.7)  parabola (-0.5,0.2);
					\draw[->,>=stealth', dashed, line width=0.4mm] (-2.5,+2.5)  parabola (-0.2,0.4);
					
					\draw[->, line width=0.18mm] (0,0) --  (0.5,0.84);
					\draw[->, line width=0.18mm] (0,0)  -- (2,1.5);
					\node[scale=1] at (0.2, 0.6) {$r$};
					\node[scale=1] at (0.8, 1.45) {$r+vt\approx R$}; 
					\node[scale=1.1] at (-0.3, -0.3) {$B_r$};
					\node[scale=1.1] at (-1., -1.3) {$B_{r+vt}$};
					
				\end{tikzpicture}

				\caption{Explanation of the main  result (Theorem \ref{pMVB main}): Taking $R\approx r+vt$, we obtain  that only particles that were initially in the enlarged ball $B_{r+vt}$ can end up in $B_r$ after time $t$ (solid arrow). Particle transport from distances $\gg vt$ is suppressed (dashed arrow).} \label{fig explaining main result}
			\end{figure}
			
			The key constraint is $0\le{t}<(R-r)/v$. This can be rephrased as $R>r+vt$ and in most applications we take $R\approx r+vt$. The velocity bound $v$ is essentially $\kappa$ from \eqref{kappa}. Note that this is a natural bound (via the Schur test) on the matrix norm $\|[J,|x|]\|$. Physically, we can interpret $[J,|x|]$ as the velocity operator of a quantum particle whose kinetic energy is described by the $|\Lambda|\times |\Lambda|$ matrix $J=(J_{xy})$ and the norm of the velocity operator is then a natural one-body speed bound. In particular, the bound is \textit{optimal} in the non-interacting case (where $V=0$). As in other recent works (\cites{MR4419446, MR4470244, SiZh,LRSZ,osborne2023locality}), this shows that for bosonic quantum many-body Hamiltonians of the form \eqref{H}, particle transport can be controlled by the hopping part. {The condition \eqref{alCond} on $\alpha$ can probably be improved, but this will require new ideas.}


			

			\subsection{Literature review and discussion}\label{sec:rr}
			
			The derivation of propagation bounds for quantum many-body Hamiltonians on lattices is a hot topic. The last 20 years have seen critical advances in the derivation and application of quantum many-body propagation bounds known as Lieb-Robinson bounds (\cite{lieb1972finite}). Lieb-Robinson bounds establish the existence of a maximal speed of quantum propagation for all local observables in quantum spin systems and lattice fermions. As first discovered by Hastings, Lieb-Robinson bounds are among the few robust mathematical tools for resolving longstanding problems in mathematical physics. Early examples are \cite{hastings2004lieb,hastings2005quasiadiabatic,hastings2006spectral,bravyi2006lieb,nachtergaele2006lieb,nachtergaele2006propagation,hastings2007area} and many other extensions and applications of Lieb-Robinson bounds followed, cf. the reviews \cite{nachtergaele2010lieb,gogolin2016equilibration,chen2023speed}. As the understanding of finite- and short-range problems advances, more attention is directed to deriving useful propagation estimates for quantum many-body systems with long-range interactions which are more realistic in practice \cite{CGS,Fossetal,TranEtAl1,TranEtAl2,TranEtal4,KS_he,LRSZ,vanvu2023optimal}.
			We refer to \cite{DefenuEtAl} and the references therein for a more comprehensive review of the effect of long-range interactions  on the transmission of quantum information and to \cite{Bose} for a more introductory overview on the subject matter.
			Importantly both quantum spin systems and lattice fermions enjoy \textit{bounded} interactions. 
			
			By contrast, the mathematics of propagation bounds for \textit{bosonic} systems (that we consider here) has been considerably less developed because these have inherently \textit{unbounded} interactions. This leads the standard  method for proving propagation bounds (which are based on operator norm estimates) to break down. Over the years, several influential works considered special instances of bosonic systems and initial states (\cites{cramer2008locality,cramer2008exact,nachtergaele2009lieb,eisert2009supersonic,SHOE,junemann2013lieb,woods2015simulating,woods2016dynamical,wang2020tightening}). For example, the authors of the seminal 2011 work \cite{SHOE} obtained  exponential bounds on the propagation admitted by Bose-Hubbard Hamiltonians into the regions of space that are initially devoid of particles, a special situation that is relevant to releasing trapped particles. The bound comes with a prefactor given by the total number of particles $N$ and is therefore not stable in the thermodynamic limit where $N\to \infty$.  A bosonic Lieb-Robinson bound with velocity scaling as $\sqrt{N}$ was derived in 2020 (\cite{wang2020tightening}).
			The last two years saw rapid progress in the area of bosonic propagation bounds and the resolution of several longstanding problems (\cite{YL,KS2,MR4419446,MR4470244,LRSZ,KVS,vanvu2023optimal}).  

			{A \textit{microscopic} bound on boson transport --- such as our main result \eqref{MVB0} --- controls  the flow of particles on $\mathcal O(1)$ length scales, ideally starting from rather general initial states.} ({Another type of bound concerns macroscopic transport  \cite{MR4419446,van2023topological,vanvu2023optimal}, which is a rougher concept that we briefly review in Remark \ref{rmk:macro} for context.) 
   
   On the microscopic transport problem there have been several recent advances based on new ideas and methodology. 
			
			\begin{enumerate}[label=(\roman*)]
				\item \cite{MR4470244} bounded the maximal speed for the first moment ($p=1$) in general initial states. However, the error term depends on the total particle number and so the bound is not stable in the thermodynamic limit $N\to \infty$. An equivalent perspective is that the bounds are effective only on slightly mesoscopic length scales $\sim N^{\eta}$ with $\eta=\eta(\alpha)>0$.
				\item \cite{KVS} proved that the maximal speed for all moments ($p\geq 1$) is almost bounded (grows at most logarithmically in time) in the case of nearest-neighbour hopping and for general initial states. 
				\item \cite{LRSZ} proved that the maximal speed for the first moment ($p=1$) is bounded for any finite-range hopping and bounded-density initial states.  
			\end{enumerate}
			
			The speed bound in (i) only holds on mesoscopic scales $r\sim N^{\eta(\alpha)}$. The speed bounds in (ii) \& (iii) hold down to microscopic length scales, but they only cover finite-range interactions. {Bosonic Hamiltonians with long-range hopping combine the mathematical challenges coming from (a) the unboundedness and (b) the long-range nature. They are a natural frontier in our understanding of quantum propagation.} \textit{Our result is the first microscopic particle propagation bound that holds for long-range hopping.}

			To appreciate the relevance of our main result and how it fits into the recent literature, we summarize recent works on particle propagation bounds for lattice bosons --- and how our result fits into the literature --- in the following table.

			\begin{table}[H]
				\centering
				{\begin{tabular}{|c|c|c|c|}
						\hline
						Precision &Hopping range&Initial state&Ref.\\
						\hline
						macroscopic&long-range ($\alpha>d+2$)&general&\cite{MR4419446}\\
						&long-range ($\alpha>d+1$)&general&\cite{vanvu2023optimal}\\
						\hline
						mesoscopic &long-range ($\alpha>d+2$)& general &\cite{MR4470244}  \\
						\hline
						&nearest neighbour&$e^{-\mu N}$ &\cite{YL}\\
						microscopic&nearest neighbour& general & \cite{KVS} \\
						&finite-range & bounded density & \cite{LRSZ}\\ 
						&long-range ($\alpha>2d+1$)& bounded density & {this work} \\
						\hline
				\end{tabular}}
				\label{table}
				\caption{Overview of bounds proved on particle transport for bosonic quantum many-body dynamics in the last two years.}
			\end{table}
			
			Some comments about Table \ref{table} are in order. 
			
			\begin{enumerate}
				\item[(i)] While results in \cite{MR4419446,MR4470244} are stated for the Bose-Hubbard model, the proof of the maximal velocity bound applies \textit{verbatim} for arbitrary density-density interactions, as long as the hopping terms in the Hamiltonian have sufficiently fast polynomial decay (c.f.~\cite[Thm.~2.1, 2.4]{SiZh}). 
				\item[(ii)] This table only concerns bounds on the transport of particles, which is arguably the most fundamental form of quantum transport. More general bounds on propagation of quantum information are \textit{Lieb-Robinson bounds} which control the spreading of general local quantum observables. (E.g.\ experiments \cite{cheneau2012light} track the evolution of the boson parity operator $e^{i\pi n_x}$; the particle number bounds are insensitive to this.)  The recent {mathematical progress on bosonic propagation bounds also led to the first general bosonic Lieb-Robinson bounds, but} the bounds are weaker and depend in various ways on the initial state \cite{YL,MR4470244,LRSZ,KVS}. {Our result opens the door to improving and extending these results for long-range bosonic Hamiltonians; see Section \ref{sect:future}.}
			\end{enumerate}

			   \begin{remark}\label{rmk:macro}
   Macroscopic transport bounds are a rougher concept which was the content of \cite{MR4419446,vanvu2023optimal}, but not of this paper. They are essentially of the form
				\begin{equation}\label{eq:macro}
					P_{N_X\geq (1-\theta) N} e^{i\mathrm{t} H} P_{N_Y\geq \theta N} e^{-i\mathrm{t} H} P_{N_X\geq (1-\theta) N}\lesssim d_{XY}^{-n},\qquad 0<t<\frac{d_{XY}}{v}
				\end{equation}
				where $P_{N_X\geq \lambda}$ is the spectral projector of $N_X$ onto eigenvalues $\geq\lambda$, $0<\theta<\theta'<1$ and $v>\kappa$ are parameters, and $d_{XY}$ the distance between regions $X$ and $Y$. The macroscopic transport bound \eqref{eq:macro} says that it takes time proportional to $d_{XY}$ to move a macroscopic fraction $\theta'-\theta$ of the total boson number $N$ from regions $X$ to $Y$. In other words, it bounds the collective speed of large ``clouds''  of bosons.

				Such finite speed bounds on macroscopic transport were derived for general initial states under the condition $\alpha>d+2$ in \cite{MR4419446} which was improved to $\alpha>d+1$ in \cite{vanvu2023optimal}. These bounds are thermodynamically stable. (We call a bound thermodynamically stable if it is independent of the particle number $N$ and the lattice volume $\Lambda$.) The decay condition $\alpha>d+1$ is argued to be sharp (\cite{vanvu2023optimal}) and so the comparatively rough notion of \textit{macroscopic} particle transport for long-range bosons is basically fully understood now. 
    \end{remark}

			\subsection{Proof strategy and key challenges}\label{sec:proofstrategy}
			The centerpiece of our proof are approximate monotonicity estimates for certain \textit{adiabatic spacetime localization observables} (ASTLOs), which dynamically track the relevant local particle propagation. The approximate monotonicity estimate is derived by microlocal methods (in particular resolvent-based commutator expansion) where the role of the small parameter is played by the inverse of the radius of the large ball, $1/R$ (which we can essentially think of as $ 1/(vt)$).
			
			The ASTLO method was developed for bosonic quantum many-body systems in \cites{MR4419446, MR4470244, SiZh,LRSZ}. It is inspired by propagation estimates introduced in \cite{SigSof2} for few-body quantum mechanics in the continuum, further developed in \cite{Skib, HeSk, BoFauSig,APSS,BFLS, BFLOSZ,hinrichs2023lieb} and also recently applied to the nonlinear Hartree equation \cite{arbunich2023maximal}. Roughly speaking, ASTLOs monotonously decrease along the Heisenberg dynamic dual to \eqref{SE}, up to explicit time-decaying remainders. Moreover, ASTLOs have suitable geometric properties that render them comparable with the number operator associated with the propagation regions of interest. This, together with the approximate monotonicity, allows us to control the spacetime localization properties of the evolving system in the Heisenberg picture. 
			
			In proving our main result we have overcome two new technical challenges that were not treated in \cites{MR4419446, MR4470244, SiZh,LRSZ}. First, to obtain a thermodynamically stable bound that does not use any a priori bounds on particle numbers appearing in various remainder terms in the commutator expansion, we have to dynamically handle the contribution of particles from very far away. Specifically, to obtain the first moment bound \eqref{MVB0}, we resort to a multiscale induction, proceeding from large scales down to $\mathcal O(1)$ length scales. Note that this multiscale scheme proceeds downwards, from large to small length scales. The induction base at large  length scales are the propagation estimates found in our previous works (\cites{MR4470244, SiZh}) because these are thermodynamically stable on sufficiently large length scales.  This multiscale rendition of the ASTLO method is a main technical contribution in our work. 
			
			Our second goal in this work is to generalize \eqref{MVB0} to higher moments (see \eqref{MVBp}), which are relevant for applications via Markov's inequality. To this end, we have designed new ASTLOs of order $p\ge1$, which are comparable to localized number operators raised to the $p$-th power. To establish the approximate monotonicity for these higher-order ASTLOs, we control the Heisenberg derivative of the $(p+1)$-th order ASTLOs by those with order $p$ up to various commutator terms which we prove to be of lower order. Thus, we nest another upward induction over the moments into the downward multiscale scheme over length scales, thereby establishing the main result,  estimate \eqref{MVBp}, for all moments $p\geq 1$.
			
			\subsection{Future directions}\label{sect:future}
			Noticee that the condition on the power-law decay exponent $\alpha$ for which we obtain a bounded speed for the first moment is $\alpha>2d+1$. The threshold $\alpha>2d+1$ appears in the context of quantum information transport of long-range interacting systems: it has been identified as the sharp threshold for having a Lieb-Robinson bound with linear light cone (which is another way to say bounded speed) in \cites{KS1,KS2,EldredgeEtAl,TranEtal4,TranEtal5}. It is notably different from the threshold $\alpha>d+1$ for bounded speed that has been recently identified for \textit{macroscopic} particle transport \cite{MR4419446,vanvu2023optimal} (which we recall is a rougher way to track transport more in the context of statistical physics) and in \cite{MR4470244}. These results suggest that there could be fundamental qualitative differences between transport on microscopic versus macroscopic length scales in long-range quantum many-body systems and we plan to investigate this in future work.
			
			We close with another comment about the more general Lieb-Robinson bounds (LRBs). It turns out that microscopic particle propagation bounds are a key ingredient to prove LRBs for bosonic Hamiltonians of the form \eqref{H} \cite{MR4470244,KVS}. The rough idea is that, as mentioned above, the main obstacle to prove an LRB for bosons are large local interactions which are due to local accumulation of bosons. This accumulation is precisely what can be controlled by particle propagation bounds. Since our main result is the first thermodynamically stable microscopic particle propagation bound for long-range bosons, it opens the door to deriving the first thermodynamically stable LRBs for long-range bosons by modifying the blueprint (constructions of suitable ``truncated dynamics'') created in \cite{MR4470244,KVS} {for localized, respectively, finite-density initial states.}
			
			\section{Preliminaries}\label{secPrelim}
			
			In this section, we present some basic definitions and technical lemmas.
			
			\subsection{Function class $\mathcal{E}$}\label{secE}
			
			Take $\eps>0$ to be determined later. 
			We define the function class
			\begin{equation}\label{classE}
				\mathcal{E}\equiv \cE_\epsilon:=\Big\{f\in C^{\infty} (\mathbb{R})\: : \: f\geq0,\: f'\geq 0,\:\sqrt{f'}\in C^{\infty} (\mathbb{R}), \:  f\equiv 0 \text{ on } (-\infty,\epsilon/2),\: f\equiv 1 \text{ on } (\epsilon,\infty)\Big\}.
			\end{equation}
			Essentially, elements in $\cE$ are {cutoff} functions with compactly supported derivatives, see \figref{figF} below.
			Later on, we will specify the value of $\eps$ according to the values of $v$ entering the statement of \thmref{pMVB main}.
			
			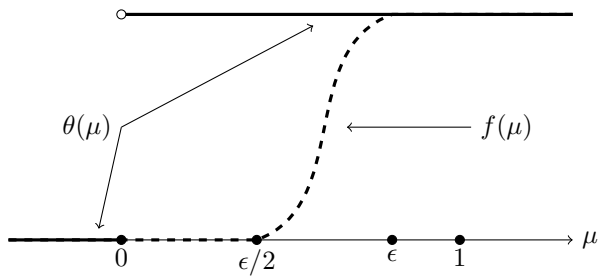
\begin{figure}[H]
				\centering
				\begin{tikzpicture}[scale=3]
					\draw [->] (-.5,0)--(2,0);
					\node [right] at (2,0) {$\mu$};
					\node [below] at (0,0) {$0$};
					\draw [fill] (0,0) circle [radius=0.02];
					
					\node [below] at (.6,0) {$\eps/2$};
					\draw [fill] (.6,0) circle [radius=0.02];
					
					\node [below] at (1.2,0) {$\eps$};
					\draw [fill] (1.2,0) circle [radius=0.02];
					
					\node [below] at (1.5,0) {$1 $};
					\draw [fill] (1.5,0) circle [radius=0.02];

					\draw [very thick] (-.5,0)--(0,0);
					\draw [very thick] (0,1)--(2,1);				
					\filldraw [fill=white] (0,1) circle [radius=0.02];
					
					\draw [dashed, very thick] (-.5,0)--(.6,0) [out=20, in=-160] to (1.2,1)--(2,1);

					\draw [->] (1.55,.5)--(1,.5);
					\node [right] at (1.55,.5) {$f(\mu)$};
					
					\draw [->] (0,.5)--(.85,.95);
					\draw [->] (0,.5)--(-.1,.05);
					\node [left] at (0,.5) {$\theta(\mu)$};
				\end{tikzpicture}
				\caption{A typical function $f\in \cE$ compared with the Heaviside function  $\theta(\mu)$.}\label{figF}
			\end{figure}
			
			Functions in $\cE$ are easy to construct. Indeed, for any $h\in C^\infty(\Rb)$ with $h\ge0$ and ${\supp h\subset (\eps/2,0, \eps)}$, let $f_1(\lam):=\int_{-\infty}^\lam h(s)ds$. Then  $f_1/\int_{-\infty}^\infty h(s)ds\in \cE$, and so we have $h \le C f'$ for some $f\in\cE$.
			
			Similarly, one could also check that 
			\begin{align}
				\label{E1}
				\text{If $f_1\in\cE$, then there exists $f_2\in\cE$ with $f_1'\le Cf_2'$. }
			\end{align}
			Moreover, 
			\begin{align}
				\label{E2}
				\text{If $f_1,f_2\in\cE$, then $f_1+f_2\le C f_3$ for some $f_3\in\cE$.}
			\end{align}

			We will frequently use the following symmetrized expansion formula for functions in $\cE$, whose proof is found in \cite{MR4470244}:
			\begin{lemma}[\cite{MR4470244}*{Lem.~2.2}]\label{lemSymExp}
				Let $n\ge1$ be an integer and $f\in\mathcal{E}$. Then, with $u:=(f')^{\frac{1}{2}}$ and, for $n\ge2$, functions $j_k\in\mathcal{E}$, $\Tilde{u}_k:=(j_k')^{\frac{1}{2}}$, $2\leq k\leq n$, there exist positive constants $C_{f,k}$ such that  for  all $x,y\in\mathbbm{R}$,
				\begin{equation}\label{eqSymExp}
					f(x)-f(y)=(x-y)u(x)u(y)+\sum_{k=2}^n (x-y)^k h_k (x,y)+R_n(x-y)^{n+1},
				\end{equation}
				where the sum should be dropped for $n=1$ and
				\begin{align}\label{hkEst}
					|h_k(x,y)|\leq& C_{f,k} \Tilde{u}_k(x)\Tilde{u}_k(y)\qquad (2\leq k\leq n),\\
					\label{RnEst}
					\abs{R_n(x,y)}\le& C_{f,n}.
				\end{align}
			\end{lemma}

			\subsection{Commutator expansion}
			Fix a lattice domain $\Lam\subset\cL$. For any function $g\in \ell^\infty(\Lam)$, the second quantization map, $\dG$,  is given by
			\begin{align}
				\label{dGDef}
				\dG (g):=\sum_{x\in\Lam}g(x)n_x,\quad n_x=a_x^*a_x.
			\end{align}
			The following commutator expansion formula is a consequence of the canonical commutator relation:
			\begin{lemma}[c.f.~\cite{MR4470244}*{Lem.~A.2}]\label{lem:HRf-com}
				Let $H_\Lam$ be as in \eqref{H}. Let $g\in\ell^\infty(\Lam)$. In the sense of forms on $\mathcal{D}(H_\Lam)\cap\mathcal{D}(N_\Lam)$, we have
				
				\begin{equation}\label{HRf-com}
					\sbr{H_\Lam,\dG(g)}=-\sum_{x,y\in\Lambda}J_{xy}\del{g(x)-g(y)}a_x^{*}a_y.
				\end{equation}
			\end{lemma}
			\begin{proof} 
				By definition \eqref{dGDef},  $\mathrm{d}\Gamma(g)$ commutes with any $n_z$ and therefore with arbitrary functions of the bosonic number operators $V=\Phi(\Set{n_z}_{z\in\Lam})$ (see \eqref{H}). Thus we have
				\begin{align} \label{H0Rf-com}
					[H_\Lam,\mathrm{d}\Gamma(g)] = \sum_{x\in\Lambda,y\in\Lambda} J_{xy} [a_x^*a_y, \dG(g) ]&=\sum_{x\in\Lambda,y\in\Lambda}\sum_{z\in\Lambda}J_{xy}  g(z)[a_x^*a_y, a_z^*a_z].\end{align}
				By the canonical commutation relation $[a_x,a_x^*]=\delta_{ij}$, we have $[a_x^*a_y, a_z^*a_z]=-a_x^*a_y$ for $z=x$, $=a_x^*a_y$ for $z=y$, and $=0$ elsewhere. 
				These facts, together with expression \eqref{H0Rf-com}, give \eqref{HRf-com}.
			\end{proof}
			

			
			\subsection{ASTLOs --- adiabatic spacetime localization observables}
			Let $v>\kappa$ with $\kappa$ defined in \eqref{kappa}, and $v':=\frac12(\kappa+v)$. For any function $f\in L^\infty(\Rb)$, $t\in\Rb$, and $R,\,s>0$, we define 
			\begin{equation}\label{ftsDef}
				f_{ts}(x)\equiv f(|x|_{ts}):=f\left(\frac{R-v't-|x|}{s}\right).
			\end{equation}
			Of particular interest are  the time-dependent observables
			\begin{align}\label{propag-obs1}
				N_{f,ts}:=\dG(f_{ts}), \quad  f \in \cE.
			\end{align}
			Following \cites{MR4419446,MR4470244,LRSZ,SiZh}, we call \eqref{propag-obs1} adiabatic spacetime localization observables (ASTLOs). As we will see later, 
			control over the evolution of ASTLOs grants the same over the spacetime localization properties of evolving states. 
			
			
			\section{Proof of \thmref{pMVB main} for $p=1$}\label{sec proof MVB}
			
			In this section, we prove our main result,  \thmref{pMVB main}, for the first moment case.

			\begin{theorem}[Thermodynamically stable propagation estimates]\label{MVBmain}
				Let \eqref{Jcond} hold with ${\al>2d+1}$. Then, for every $v > \kappa$ {and $\delta_0>0$, there exists $C=C(\al,d, C_J, v,\delta_0)>0$ such that for all $\l,\,R,\,r>0$ with $R-r>\max(\delta_0r,1)$  and initial states $\psi_0\in\cD(N^{1/2}_\Lam)$ satisfying \eqref{UDBp} with $p=1$,}
				\begin{align}
					\label{MVB1}
					\sup_{0\le t <(R-r)/v}\br{N_{B_r}}_t \le (1+C(R-r)^{-1})\br{N_{B_{R}}}_0+C\l.
				\end{align}
			\end{theorem}
			\thmref{MVBmain} is proved at the end of this section.

			The starting point of our proof strategy, following \cites{MR4419446,MR4470244,LRSZ,SiZh}, is to derive monotonicity estimates for the ASTLOs, \eqref{propag-obs1}. Compared to our previous results, in Props.~\ref{lemLocIntRME}--\ref{propMVB1} below, we refine the monotonicity estimates to nail down the localization properties of the error terms. This enables us to subsequently use a multiple-scale argument in a backward induction scheme to {consecutively treat the contribution of particles at various length scales and thus	remove the dependence on total particle number in the remainder terms. See the proof of Proposition \ref{propMultScale} for details and  \figref{fig1} below for an illustration.}
			
			\begin{figure}[H]
				\centering
				
				\begin{tikzpicture}[scale=.8]

					\draw[black,line width=0.2mm]
					(0,0) circle (1.5)
					(0,0) circle (3.);
					
					\draw[->,>=stealth', line width=0.4mm] (1.8,1)  parabola (0.2,0.2);
					\draw[->,>=stealth', dashed, line width=0.4mm] (3.5,-2)  parabola (0.2,-0.5);
					
					\node[scale=1.1] at (-1.3, 1.3) {$2^{k}$};
					\node[scale=1.1] at (-2.43, 2.43) {$2^{k+1}$};
					\node[scale=1.1] at (-1.3, 1.3) {$2^{k}$};
					\node[scale=1.1] at (-0.7, -0.9) {I};
					\node[scale=1.1] at (-1.2, -2) {II};
					\node[scale=1.1] at (-2, -3) {III};
					
				\end{tikzpicture}

				\caption{In the proof of \thmref{MVBmain}, at each induction step we control the particles moving from region II to region I, as the contribution from region III into region I have been already accounted for by the induction hypothesis.  } \label{fig1}
			\end{figure}
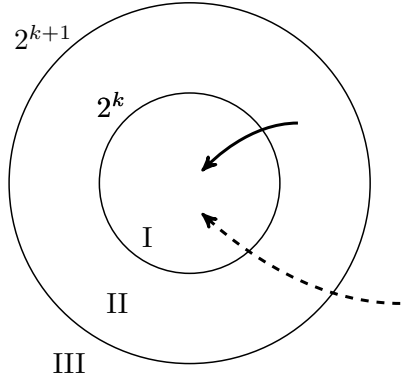

			To begin with, from Lems.~\ref{lemSymExp}--\ref{lem:HRf-com} we derive the following proposition. 
			\begin{proposition}\label{lemLocIntRME} Let \eqref{Kcond} hold for $n\ge1$.
				Then, for any $v>\kappa$,  $f\in\mathcal{E}$, there exist ${C=C(f,n,v)>0}$ and, for $n\ge2$, functions $j_k\in\mathcal{E}$, $2\leq k\leq n$, such that for all $t,s,{R}>0$,
				\begin{align}\label{intRME}
					\int_0^t \langle N_{f', \tau s }\rangle_{\tau}d\tau
					\leq& C\del{ s \langle N_{f, 0s }\rangle_0+\sum_{k=2}^ns^{-k+2}\langle N_{j_k,0s}\rangle_0+s^{-n} t\Rem(t)},\\
					\Rem(t):=&\sup_{t'\le t}\del{\langle N_{B_{ {R}-\eps s/2}}\rangle_{t'}+\sum_{x\in B_{{R}-\eps s/2}}\sum_{y\in\Lambda}\abs{J_{xy}}|x-y|^{n+1}\langle n_y\rangle_{t'} }.\label{RemDef}
				\end{align} 
				The sum in \eqref{intRME} is dropped for $n=1$.
			\end{proposition}
			\propref{lemLocIntRME} is proved in \secref{secPf41}.
			
			Next,	by adapting the arguments in \cite[Sect.~2.2]{MR4470244}, \cite[Sect.~4.2]{SiZh}, we prove in \secref{secPf42} that
			\propref{lemLocIntRME} implies the following propagation estimate:
			\begin{proposition}\label{propMVB1} Let \eqref{Kcond} hold for $n\ge1$. 
				Then, for any $v>\kappa$, there exists ${C=C(n, v)>0}$ such that for all ${R}>r>0$ and $s=({R}-r)/v$,
				\begin{align}\label{propEst1}
					\sup_{t\le s}	\br{N_{B_r}}_t\le (1+Cs^{-1}){\br{N_{B_{R}}}_0+Cs^{-n}\Rem(s)}.
				\end{align}
				Here $\Rem(s)$ is given by \eqref{RemDef}.
			\end{proposition}
			\begin{remark}
				Note that, unlike estimate \eqref{intRME}, the auxiliary functions $f\in\cE$ etc.~do not enter the statement of \propref{propMVB1}.
			\end{remark}
			
			Estimate \eqref{propEst1} forms the basis of the proof of \thmref{MVBmain}.
			In the next proposition, we illustrate the iterative scheme with the choice ${R}=2^{k+1}$ and  $r=2^k$ for any $k>0$. The corresponding result for general ${R},\,r$ can then be obtained by straightforward adaption and is deferred to the end of this section.

			\begin{proposition}\label{propMultScale}
				Let \eqref{Jcond} hold with $\al>2d+1$. Then, for any $v > \kappa$, there exists some ${C=C(\al, n, C_J, v)>0}$ such that for all $\l,\,k\ge0$, and initial states with $\sup_{x\in\Lam}\br{n_x}_0\le \l$,
				\begin{align}
					\label{multEst}
					\sup_{vt\le 2^k}\br{N_{B_{2^k}}}_t\le (1+2^{-k}C) { {\br{N_{B_{2^{k+1}}}}_0+ {C\l} }}.
				\end{align}
			\end{proposition}

			\begin{proof}
				Take a large constant $L>0$ s.th.~the domain $\Lam\subset [-L/2,L/2]^d$. Let $C_0>0$ be some constant independent of $a,\,L,\,k$ to be determined later.
				We prove \eqref{multEst} with any $C\ge C_0$ by a downscale induction on the size of $L$.

				1. First,	for the base case, we prove that there exists a large integer $K=K(L,\al,d)>0$ s.th. \eqref{multEst} holds for all $k\ge K$. 
				
				Let 
				\begin{align}\label{nDef}
					n:=
						\lfloor\al-d-1\rfloor.
					\end{align} Then, by the assumption $\al>2d+1$, condition \eqref{Kcond} holds with $n\ge d\ge1$.
					
					Thus, we can apply {\cite[eq.~(2.9)]{SiZh} (see also \cite[eq.~(7)]{LRSZ})}
					to obtain, for some $C_1>0$\textit{ independent of $k,\,\alpha,\,L$},
					\begin{align}\label{propEst3}
						\sup_{vt\le 2^k}\langle N_{B_{2^k}}\rangle_{t}\le (1+2^{-k}C_1) {\br{N_{B_{2^{k+1}}}}_0 + 2^{-kn}C_1\br{N_\Lam}_0}.
					\end{align}
					Under the assumptions $\sup_x\br{n_x}_0\le \l$ and $\abs{\Lam}\le L^d$, we have $\br{N_\Lam}_0\le \l L^d$, and so estimate \eqref{propEst3} becomes
					\begin{align}\label{propEst4}
						\sup_{vt\le 2^k}\langle N_{B_{2^{k}}}\rangle_{t}\le (1+2^{-k}C_1) {\br{N_{B_{2^{k+1}}}}_0 +  2^{-kn} L^d}C_1\l.
					\end{align}

					We conclude from \eqref{propEst4} that \eqref{multEst} holds for $C\ge C_1$, $k\ge K$ with 
					\begin{align}\label{Kchoice}
						K=\frac{d}{n}\log_2(L).
					\end{align} 
					This completes the proof of the base case.\\
					
					2. Next, assuming \eqref{multEst} holds for $k+1$, we prove it for $k$. 
					
					Recall definition \eqref{RemDef} for $\Rem(t)$. We apply \propref{propMVB1} with ${R}=2^{k+1}$ and $r=2^k$, whence $s=2^k/v$ and \eqref{propEst1} becomes
					\begin{align}\label{propEst2}
						\sup_{vt\le 2^k}	\br{N_{B_{2^k}}}_t\le&(1+2^{-k}C_2){\br{N_{B_{2^{k+1}}}}_0+2^{-nk}C_2\sup_{vt\le 2^k}\langle N_{B_{2^{k+1}}}\rangle_{t}}
						\notag\\
						&+2^{-nk}C_2\sup_{vt\le 2^k}\sum_{x\in B_{2^{k+1}}}\sum_{y\in\Lambda}\abs{J_{xy}}|x-y|^{n+1}\langle n_y\rangle_{t}.
					\end{align}
					The first term in the r.h.s.~of \eqref{propEst2} {is of the desired form} with prefactor bounded by $C_1$. To bound the second term, we claim that, for some dimensional constant $C_d>0$, 
					\begin{align}
						\label{propEst6}
						2^{-dk}\sup_{vt\le 2^k}\langle N_{B_{2^{k+1}}}\rangle_{t}\le C_d\,\l.
					\end{align}
					Indeed, suppose \eqref{propEst6} holds. Then,  since {$n\ge d$ by the assumption $\al>2d+1$} (see \eqref{nDef}), we have $2^{-nk}C_2\sup_{vt\le 2^k}\langle N_{B_{2^{k+1}}}\rangle_{t}\le C_d C_2\,\l$, as desired.
					
					Now we prove \eqref{propEst6}. 
					We compute,  using the induction hypothesis {and the volume growth condition \eqref{growthCond1},} that
					\begin{align}
						\label{propEst5}
						2^{-dk}\sup_{vt\le 2^k}\langle N_{B_{2^{k+1}}}\rangle_{t}\le& 2^{-dk}\del{\del{1+C_02^{-{(k+1)}}}\br{N_{B_{2^{k+2}}}}_0+C_0\l}\notag\\
						\le&\del{1+C_02^{-{(k+1)}}}V_d 2^{2d}\,\l+2^{-dk}C_0\l .
					\end{align}
					Here $V_d$ is the dimensional constant entering \eqref{growthCond1}. Since $C_0,\,C_2$ are both independent of $K$, increasing $K$ in \eqref{Kchoice} if necessary, we have for all $k> K$ that
					\begin{align}
						\label{414}
						1+C_02^{-k}\le 11/10, \quad 2^{-dk}C_0 \le 1/10.
					\end{align}
					Take  $C_d:=\frac{11}{10}V_d\cdot 2^{2d} +\frac{1}{10}$ for all $d$ to conclude from \eqref{propEst5} that \eqref{propEst6} holds.
					
					3.	It remains to bound the term appearing in line \eqref{propEst2}. 
					
					We write
					\begin{align}
						&2^{-nk}C_2\sup_{vt\le 2^k}\sum_{x\in B_{2^{k+1}}}\sum_{y\in\Lambda}\abs{J_{xy}}|x-y|^{n+1}\langle n_y\rangle_{t}\le \mathrm{I}+\mathrm{II},\label{12}
					\end{align}
					where
					\begin{align}
						\mathrm{I}:=&2^{-nk}C_2\sup_{vt\le 2^k}\sum_{\abs{x}\le 2^{k+1}}\sum_{\abs{y}\le 2^{k+1}}\abs{J_{xy}}|x-y|^{n+1}\langle n_y\rangle_{t},\label{Idef}\\ 
						\mathrm{II}:=&2^{-nk}C_2\sup_{vt\le 2^k}\sum_{\abs{x}\le 2^{k+1}} \sum_{l> k+1}   \sum_{2^l<\abs{y}\le 2^{l+1}} \abs{J_{xy}}|x-y|^{n+1}\langle n_y\rangle_{t}\label{IIdef}.
					\end{align}
					
					Owning to condition \eqref{Kcond}, the term $I$ in \eqref{12} can be bounded as 
					\begin{align}\label{Iest1}
						\Mr{I}\le&2^{-nk}C_2\del{\sup_{x\in\Lam}\sum_{y\in\Lam} \abs{J_{xy}}\abs{x-y}^{n+1}}
						\del{\sup_{vt\le 2^k}\sum_{\abs{y}\le 2^{k+1}}\langle n_y\rangle_{t}}\notag
						\\=& \kappa_n2^{-nk}C_2\sup_{vt\le 2^k}\br{N_{B_{2^{k+1}}}}_t.
					\end{align}

					Applying \eqref{propEst6} to \eqref{Iest1}, we find that for $n\ge d $, 
					\begin{align}\label{Iest2'}
						\Mr{I}	\le& \kappa_n C_dC_2\l.
					\end{align}
					This completes the  bound for the term $\Mr{I}$ in \eqref{Idef}.

					Next, set \begin{align}
						\label{gDef}
						\g:=\al-n-1.
					\end{align}
					Using growth condition \eqref{growthCond2}, the decay condition \eqref{Jcond}, and the reverse triangle inequality $\abs{x-y}\ge \abs{\abs{x}-\abs{y}}$, we bound  second term $\Mr{II}$ in \eqref{12}  as 
					\begin{align}
						\Mr{II}\le& 2^{-nk}C_2C_J\sup_{vt\le 2^k}\sum_{\abs{x}\le 2^{k+1}} \sum_{l> k+1} \sum_{2^l<\abs{y}\le 2^{l+1}}   \abs{x-y}^{-\g}\langle n_y\rangle_{t}\notag\\
						\le& 2^{-nk}C_2C_J \sup_{vt\le 2^k}\sum_{\abs{x}\le 2^{k+1}} \sum_{l> k+1}\sum_{2^l<\abs{y}\le 2^{l+1}}\del{2^l-\abs{x}}^{-\g} \langle n_y\rangle_{t} \notag\\
					\le & 2^{-nk}C_2C_J\om_{d-1}\sum_{m=0}^{2^{k+1}}m^{d-1} \sum_{l> k+1}\del{2^l-m}^{-\g}\sup_{vt\le 2^k}\br{N_{B_{2^{l+1}}}}_t\notag\\
					\le&\underbrace{C_2C_J2^{-(n-d)k}}_{\Mr{A}}\underbrace{\del{2^{-dk}\om_{d-1}\sum_{m=0}^{2^{k+1}}m^{d-1}}}_{\Mr{B}} \underbrace{\del{\sum_{l> k+1}2^{-(l-1)\g}\sup_{vt\le 2^k}\br{N_{B_{2^{l+1}}}}_t}}_{\Mr{C}}
					.\label{IIest1}
				\end{align}
			Here $\om_{d-1}$ is the dimensional constant entering \eqref{growthCond2}.  
			
			Since $n\ge d$, the term A is $O(1)$.  The term B can be bounded by \begin{align}
	\label{IIest11}
	\Mr{B}\le &	2^{-dk}\om_{d-1}\int_0^{2^{k+1}}m^{d-1}
	= \frac{\om_{d-1}}{d}2^{d}.
\end{align}
				To bound the term C, we write
				\begin{align}
					\label{IIest12}
					&\Mr{C}=
					\sum_{l> k+1}2^{-(l-1)\g+dl} 
					\del{2^{-dl}\sup_{vt\le 2^k}\langle N_{B_{2^{l+1}}}\rangle_{t}	}.
				\end{align}
				Using estimate \eqref{propEst6}, we have $2^{-dl}\sup_{vt\le 2^k}\langle N_{B_{2^{l+1}}}\rangle_{t}\le C_d\l$ uniformly for all $l>k+1$.  By the definitions of $n,\,\g$ in \eqref{nDef}, \eqref{gDef}, we have $\g>d$ and so $\sum_{l> k+1}2^{-(l-1)\g+dl} $ converges geometrically. Owning to these facts, we find that
				\begin{align}
					\label{IIest13}
					&\sum_{l> k+1}2^{-(l-1)\g}\sup_{vt\le 2^k}\br{N_{B_{2^{l+1}}}}_t\le \frac{C_d\cdot 2^{\g}}{1-2^{d-\g}} \l.
				\end{align}
				Combining \eqref{IIest11} and \eqref{IIest13} in \eqref{IIest1}, and again using the fact that $n\ge d$, we conclude 
				\begin{align}
					\label{IIest2}
					\Mr{II} \le\frac{C_J\om_{d-1}}{d}\times \frac{C_d\cdot 2^{d+\g}}{1-2^{d-\g}} C_2\l.
				\end{align}
				This, together with the estimate \eqref{Iest2'}, yields the desired upper bound for the term in line \eqref{propEst2}.

				4. Finally, we set
				\begin{align}
					\label{}
					C_0:= \max\Set{C_1,(1+\kappa_n)C_d C_2, \frac{C_J\om_{d-1}}{d}\times \frac{C_d\cdot 2^{d+\g}}{1-2^{d-\g}} C_2}.
				\end{align}
				Clearly, $C_0$ is independent of $a,\,k,\,L$. For any $C\ge C_0$, by estimates \eqref{propEst2}, \eqref{propEst6}, \eqref{IIest1}, and \eqref{IIest2}, we conclude that \eqref{multEst} holds for $k$. 
				
				This completes the proof of \propref{propMultScale}.	
			\end{proof}
			
			{
				\begin{proof}[Proof of \thmref{MVBmain}]
					We proceed by adapting the proof of \propref{propMultScale}. 
					Since ${R}-r>1$ by assumption, we will make a downward induction on the value of ${R}-r\ge 2^k$, $k=1,2,\ldots$. 
					
					For the base case, we use \cite[eq.~(2.9)]{SiZh}, the uniform density bound \eqref{UDBp} with $p=1$, and the volume bound $\abs{\Lam}\le L^d$ to obtain, for some $C=C(\al, n, C_J, v)>0$,
					\begin{align}\label{propEst4'}
						\sup_{ vt\le {R}-r}\langle N_{B_r}\rangle_{t}\le (1+({R}-r)^{-1}C) {\br{N_{B_{R}}}_0 +  ({R}-r)^{-n} L^d}C\l.
					\end{align}
					Thus, for $K=\frac{d}{n}\log_{{R}-r}(L)$ and all $k\ge K$, we have the desired estimate \eqref{MVB1} for ${R}-r\ge 2^k$.
					
					For the induction step, assuming \eqref{MVB1} holds for  all $R,\,r$ with $R-r\ge 2^{k+1}$, 
					we prove the following key uniform estimate: for some $C'=C'(\al, n, C_J, v,\delta_0)>0$  and all $2^{k+1}>R-r\ge 2^k$,
					\begin{align}\label{propEst6'}
						(R-r)^{-d}\sup_{vt\le R-r}\langle N_{B_R}\rangle_{t}\le C'\l.
					\end{align}
					Let $(R',r')= (3R-{2r},R)$. Then $R'-r'=2(R-r)\ge 2^{k+1}$, and so it follows from the induction hypothesis that
					\begin{align}
						\label{propEst5'}
						(R-r)^{-d}\sup_{ vt\le R-r}
						\langle N_{B_R}\rangle_{t}\le& (R-r)^{-d}
						\del{\del{1+C(R-r)^{-1}}\br{N_{B_{3R-2r}}}_0+C\l}\notag\\
						\le&C\del{(R-r)^{-d}R^d+1}\l.
					\end{align}
					Since, by assumption, $R>(1+\delta_0)r$, for $\mu:=1-\frac{1}{1+\delta_0}>0$ we have $R-r>\mu R$ and therefore $(R-r)^{-d}R^d\le \mu^{-d}$. This, together with  \eqref{propEst5'}, implies the desired estimate \eqref{propEst6'}.
					
					Estimate \eqref{propEst6'} correspond to \eqref{propEst6}, through which the rest of the induction step follows from the corresponding arguments in the proof of \propref{propMultScale}, mutatis mutandis.  \end{proof}
				
			}
			
			\subsection{Proof of \propref{lemLocIntRME}}\label{secPf41}
			In this section and the next one, we prove \propref{lemLocIntRME} and \propref{propMVB1} to conclude the proof of \thmref{MVBmain}.
			{For simplicity of notation, here and below we do not display the dependence on $\Lam$ in various operators.}
			
			1. For fixed $s>0\,,f\in\cE$, we define the ASTLOs according to \eqref{propag-obs1} as
			\begin{align}
				\label{ASTLO1}
				\Phi(t):=N_{f,ts}.
			\end{align}
			We start by computing the Heisenberg derivative, $D\Phi(t)$,
			defined by
			\begin{align}\label{HeisDerDef}
				D\Phi(t):=\di_t\Phi(t)+i[H,\Phi(t)].
			\end{align}
			The temporal derivative is easily computed as 
			\begin{align} \label{eq:deriv}
				\di_t\Phi(t)=-s^{-1}v' \, N_{f',ts}.
			\end{align}
			By Lemma \ref{lem:HRf-com}, we have, for $\abs{x}_{ts}=s^{-1}\del{R-v't-|x|}$,
			\begin{align}
				\big [H , \Phi(t) \big ] &= -\sum_{x,y\in\Lambda, x\neq y} J_{xy} \del{f(|x|_{ts})-f(|y|_{ts}) }a_x^*a_y , \label{eq:commut}
			\end{align}
			in the sense of quadratic forms on $\mathcal{D}(H)\cap\mathcal{D}(N)$. Thus, by standard density argument, we have for all $\varphi\in\mathcal{D}(N^{\frac12})$ that
			\begin{align}
				\abs{\ondel{i\big [H , \Phi(t) \big ]}}
				&\le\sum_{x,y\in\Lambda, x\neq y} | J_{xy}|\abs{f(|x|_{ts})-f(|y|_{ts})} \big|\big \langle   \varphi , a_x^*a_y \varphi \big \rangle\big| .
				\label{HeisEst1}
			\end{align}

			2. Next, we observe that by definitions \eqref{classE} and \eqref{ftsDef}, for any $f\in\cE$,
			$f(\abs{x}_{ts})\ne0$ only if
			\begin{align}\label{rel1}
				\abs{x}_{ts}\equiv	\frac{R-v't-|x|}{s}>\frac{\eps}{2}\quad
				\iff \quad  |x|< R-\frac{\eps s}{2} -v't.
			\end{align}
			Hence, $f(\abs{x}_{ts})-f(\abs{y}_{ts})\ne 0$ implies that either $x$ or $y$ lie in the set $(0,R-\eps s/2  -v't)$ and, in particular, $|x|\leq R-\eps s/2$ or $|y|\leq R-\eps s/2$. 
			Consequently, denoting by $\chi_X$ the characteristic functions for  $X\subset\Lam$, we have the localization estimate
			\begin{align}
				\label{fLocRel}
				\abs{f(\abs{x}_{ts})-f(\abs{y}_{ts})}\le \abs{f(\abs{x}_{ts})-f(\abs{y}_{ts})}\del{\chi_{B_{R-\eps s/2}}(x)+\chi_{B_{R-\eps s/2}}(y)}.
			\end{align}
			Applying Lemma \ref{lemSymExp} to the r.h.s. of  \eqref{fLocRel} and using that $\abs{|x|-|y|}\le \abs{x-y}$, we  obtain
			\begin{align}
				\abs{f(\abs{x}_{ts})-f(\abs{y}_{ts})}\le& 2\del{\frac{\abs{x-y}}{s}u(|x|_{ts})u(|y|_{ts})+  \sum_{k=2}^n \frac{\abs{x-y}^k}{s^k} h_k(|x|_{ts},|y|_{ts})} \notag\\
				+ & \frac{\abs{x-y}^{n+1}}{s^{n+1}}\abs{R_n(|x|_{ts},|y|_{ts})}\del{\chi_{B_{R-\eps s/2}}(x)+\chi_{B_{R-\eps s/2}}(y)},\label{fDifEst2}
			\end{align}
			where, recall, the sum is dropped if $n=1$, $u:=(f')^{1/2}$, and $h_k$, $R_n$ satisfy \eqref{hkEst}--\eqref{RnEst}, respectively.

			Inserting \eqref{fDifEst2} into \eqref{HeisEst1} yields:
			\begin{align}
				&\quad\ondel{i\big [H , \Phi(t) \big ]} \notag\\&\le 2 \sum_{x,y\in\Lambda, x\neq y} \abs{J_{xy}} \frac{\abs{x-y}}{s}u(|x|_{ts})u(|y|_{ts}) \abs{\ondel{a_x^*a_y} }\notag \\
				&\quad + 2\sum_{k=2}^n \sum_{x,y\in\Lambda, x\neq y} \abs{J_{xy}} \frac{\abs{x-y}^k}{s^k} \abs{h_k(|x|_{ts},|y|_{ts})} \abs{\ondel{a_x^*a_y }}\notag \\
				&\quad + \sum_{x,y\in\Lambda, x\neq y} \abs{J_{xy}} \frac{\abs{x-y}^{n+1}}{s^{n+1}}\abs{R_n(|x|_{ts},|y|_{ts})}\abs{\ondel{a_x^*a_y }}\del{\chi_{B_{R-\eps s/2}}(x)+\chi_{B_{R-\eps s/2}}(y)}. \label{HestEst2}
			\end{align}
			
			The first and second sums in the r.h.s. of \eqref{HestEst2} can be treated exactly as in \cite{MR4470244}*{Sect.~2.1}. For convenience of the readers, here we omit the derivations and record the results: for $\kappa$ as in \eqref{kappa} and, for $n\ge2$, some functions $j_k\in\cE$, $k=2,\ldots,n$,
			\begin{align}
				\quad\sum_{x,y\in\Lambda, x\neq y} | J_{xy}||x-y|u(|x|_{ts})u(|y|_{ts}) \big|\big \langle  \varphi ,  a_x ^*a_y \varphi \big \rangle\big|\le&\kappa\big\langle\varphi,N_{f',ts} \,\varphi\big\rangle,\label{1stEst}\\
				\sum_{x,y\in\Lambda, x\neq y} \abs{J_{xy}} \abs{x-y}^k \abs{h_k(|x|_{ts},|y|_{ts})} \abs{\ondel{a_x^*a_y }} |\le&C_{f,k}\big\langle\varphi,N_{j_k',ts} \,\varphi\big\rangle .\label{2ndEst}
			\end{align}
			Below we focus on the third sum in the r.h.s. \eqref{HestEst2}. 
			
			3. We first treat the term involving $\chi_{B_{R-\eps s/2}}(x)$. Using remainder estimate \eqref{RnEst} applying Cauchy-Schwartz and Young's inequalities, we find
			\begin{align}
				&\quad\sum_{x,y\in\Lambda, x\neq y} \abs{J_{xy}} {|x-y|^{n+1}}\abs{R_n(|x|_{ts},|y|_{ts})}\abs{\ondel{a_x^*a_y }}\chi_{B_{R-\eps s/2}}(x)\notag\\
				&\leq {C_{f,n}}\sum_{x\in B_{R-\eps s/2}}\sum_{y\in\Lambda}\abs{J_{xy}}|x-y|^{n+1}\ondel{n_x}^{1/2}\ondel{n_y}^{1/2}\notag\\
				&\leq {C_{f,n}}\sum_{x\in B_{R-\eps s/2}}\sum_{y\in\Lambda}\abs{J_{xy}}|x-y|^{n+1}\big(\ondel{n_x}+\ondel{n_y}\big).\label{R1Est1}
			\end{align}

			By condition \eqref{Kcond}, the first sum in line \eqref{R1Est1} can be bounded as
			\begin{align}
				&\sum_{x\in B_{R-\eps s/2}}\sum_{y\in\Lambda}\abs{J_{xy}}|x-y|^{n+1}\ondel{n_x}   \notag\\=&\sum_{x\in B_{R-\eps s/2}}\ondel{n_x}\sum_{y\in\Lambda}\abs{J_{xy}}|x-y|^{n+1}\notag\\
				\leq& \del{\sum_{x\in B_{R-\eps s/2}}\ondel{n_x}}\del{\sup_{x\in\Lam}\sum_{y\in\Lam}\abs{J_{xy}}\abs{x-y}^{n+1}}\notag\\
				\leq&  \kappa_n\ondel{ N_{B_{R-\eps s/2}}}.\label{R1Est2}
			\end{align}
			Plugging \eqref{R1Est2}
			back to  \eqref{R1Est1} yields
			\begin{align}
				&\quad\sum_{x,y\in\Lambda, x\neq y} \abs{J_{xy}} {|x-y|^{n+1}}\abs{R_n(|x|_{ts},|y|_{ts})}\abs{\ondel{a_x^*a_y }}\chi_{B_{R-\eps s/2}}(x)\notag\\ 
				&\leq {C_{f,n}}\del{\ondel{B_{R-\eps s/2}}+\sum_{x\in B_{R-\eps s/2}}\sum_{y\in\Lambda}\abs{J_{xy}}|x-y|^{n+1}\ondel{n_y} }.\label{R1Est4}
			\end{align}
			
			Since $J_{xy}=\bar J_{yx}$, the sum in line \eqref{HestEst2} involving $\chi_{B_{R-\eps s/2}}(y)$ can be treated by
			interchanging the summation indices in \eqref{R1Est4} and proceeding as above. Hence, we conclude that
			\begin{align}
				&\sum_{x,y\in\Lambda, x\neq y} \abs{J_{xy}} {\abs{x-y}^{n+1}}\abs{R_n(|x|_{ts},|y|_{ts})}\abs{\ondel{a_x^*a_y }}\del{\chi_{B_{R-\eps s/2}}(x)+\chi_{B_{R-\eps s/2}}(y)}\notag\\
				\le&
				{C_{f,n}}\left( \ondel{ N_{B_{R-\eps s/2}}}+\sum_{x\in B_{R-\eps s/2}}\sum_{y\in\Lambda}\abs{J_{xy}}|x-y|^{n+1}\ondel{ n_y} \right).\label{R1Est7}
			\end{align}
			This bounds the contribution of $R_n$ in \eqref{HestEst2}.

			4. Since the vector $\varphi\in\cD(N^{1/2})$ is arbitrary, combining \eqref{1stEst}, \eqref{2ndEst}, and \eqref{R1Est7} in \eqref{HestEst2} yields the operator inequality
			\begin{align}
				\label{ME1}
				{i\big [H , \Phi(t) \big ]}\le& \kappa s^{-1} N_{f',ts} +  C_{f,n} \sum_{k=2}^n s^{-k} N_{j_k' ,ts}   \notag\\&+ \frac{ C_{f,n}}{s^{n+1}} \left( N_{B_{R-\eps s/2}}+\sum_{x\in B_{R-\eps s/2}}\sum_{y\in\Lambda}\abs{J_{xy}}|x-y|^{n+1}{ n_y} \right).
			\end{align}
			This, together with \eqref{eq:deriv}, yields the following bound for the Heisenberg derivative:
			\begin{align}
				\label{ME2}
				D\Phi(t)\le & (\kappa-v') s^{-1} N_{f',ts} + C_{f,n}\sum_{k=2}^n s^{-k} N_{j_k' ,ts}   \notag\\&+\frac{ C_{f,n}}{s^{n+1}} \left( N_{B_{R-\eps s/2}}+\sum_{x\in B_{R-\eps s/2}}\sum_{y\in\Lambda}\abs{J_{xy}}|x-y|^{n+1}{ n_y} \right).
			\end{align}
			
			5. 
			Using \eqref{ME2}, we proceed as in \cite{MR4470244}*{Sect.~2.1} to conclude estimate \eqref{intRME}. Fix $\psi\in\mathcal{D}(H)\cap\mathcal{D}(N)$. Recall that we write  $\psi_t=e^{-iHt}\psi$
			and $\br{\cdot}_t=\wtdel{(\cdot)}$.
			
			By definition \eqref{HeisDerDef}, we have the relation
			\begin{align}\label{dt-Heis}
				&\frac{d}{dt}\left<\Phi(t)\right>_t =\lan D\Phi(t)\ran_t.\ 
			\end{align}
			By the fundamental theorem of calculus, $\lan \Phi(t)\ran_t= \lan \Phi(0)\ran_0+\int_0^t \p_{\tau}\left<\Phi(\tau)\right>_{\tau} d\tau$, and therefore we find
			\begin{align} \label{eq-basic}  
				\lan \Phi(t)\ran_t-\int_0^t \lan D\Phi(\tau)\ran_\tau d\tau= \lan \Phi(0)\ran_0.
			\end{align}
			Using estimates \eqref{dt-Heis}, \eqref{eq-basic},  \eqref{ME2}, together with the definition $\Phi(t)= N_{f,ts}$, we find
			\begin{align} \label{propag-est1} 
				&\quad \big\lan N_{f,ts}\big \ran_t+(v'-\kappa)s^{-1}\int_0^t\big \lan N_{f',\tau s} \big\ran_\tau d\tau\notag\\
				& \le \,\big \lan N_{f,0s}\big \ran_0+C_{f,n} \del{\sum_{k=2}^n s^{-k}\int_0^t \big\lan N_{j_k',\tau s}\big \ran_\tau d\tau +   s^{-n-1}\int_0^t P(\tau)\,d\tau},
			\end{align}
			where $$P(\tau):=\br{ N_{B_{R-\eps s/2}}}_\tau+\sum_{x\in B_{R-\eps s/2}}\sum_{y\in\Lambda}\abs{J_{xy}}|x-y|^{n+1}\br{ n_y}_\tau.$$
			Since $\kappa < v'$, \eqref{propag-est1} implies (after dropping $\lan N_{f,ts} \ran_t\ge0$ and multiplying by $s(v'-\kappa)^{-1}>0$) that
			\begin{align}
				\int_0^t \big \lan N_{f',\tau s} \big\ran_\tau d\tau\, 
				\le C_{f,v,n} \, &
				\del{s \big\lan N_{f,0s} \big\ran_0 + \sum_{k=2}^n s^{-k+1} \int_0^t \big\lan N_{j_k',\tau s} \big\ran_\tau d\tau+    s^{-n}\int_0^t P(\tau)\,d\tau}. \label{propag-est3} 
			\end{align}
			
			6.	Applying H\"older's inequality to the last integral in line \eqref{propag-est3}, we find
			\begin{align} \label{propag-est2} 
				&\int_0^t\big \lan N_{f',\tau s} \big\ran_\tau d\tau\le C_{f,v,n}\del{s\big \lan N_{f,0s}\big \ran_0+ \sum_{k=2}^n s^{-k+1}\int_0^t \big\lan N_{j_k',\tau s}\big \ran_\tau d\tau +  ts^{-n}\Rem(t)},
			\end{align}
			where the  sum should be dropped for $n=1$ and $\Rem(t)=\sup_{\tau\le t}P(\tau)$.
			If $n=1$, \eqref{propag-est2} gives estimate \eqref{intRME}. If $n\ge2$, applying \eqref{propag-est2} to the term $\int_0^t \lan N_{j'_2,\tau s}  \ran_\tau d\tau$ and using properties \eqref{E1}--\eqref{E2} for the function class $\cE$ (see \secref{secE}), we obtain
			\begin{align}
				\int_0^t\big \lan  N_{f',\tau s}\big \ran_\tau d\tau\, 
				\le C_{f,v,n}  &
				\del{s \big\lan  N_{f,0s} \big\ran_0 + \big\lan N_{j_2 , 0s} \big\ran_0 + \sum_{k=3}^n s^{-k+1} \int_0^t \big\lan N_{\tilde{j}'_k,\tau s}\big \ran_\tau d\tau   +  t s^{-n}\Rem(t)}, \label{propag-est33} 
			\end{align}
			for some $\tilde j_k\in\cE$. Repeating the procedure, we arrive at \eqref{intRME} for $\psi\in\mathcal{D}(H)\cap\mathcal{D}(N)$. By a standard density argument, this extends to all $\psi\in\mathcal{D}(N^{1/2})$.
			
			This completes the proof of  Proposition \ref{lemLocIntRME}.\qed
			
			\subsection{Proof of \propref{propMVB1}}\label{secPf42}
			
			1. Fix any $f\in\cE$.  Since $\supp f\subset (0,\infty)$ for any $f\in \cE$ (see \eqref{classE}), we have $\supp f\big(\frac{\cdot}{s}\big) \subset (0,\infty)$ for any $s>0$.
			This implies, for any $s,\,R>0$,
			\begin{equation}\label{3122}
				f(\mu/s)\le\theta(\mu),
			\end{equation}
			where, recall, 	$\theta$ is the Heaviside fnuction (see \figref{fig:f0}).
			Setting $\mu=R-\abs{x}$ in \eqref{3122} and recalling definition \eqref{ftsDef}, we conclude that
			\begin{align}
				\label{3122'}
				f_{0s}(x)\equiv f\del{\frac{R-\abs{x}}{s}}\le\theta(R-\abs{x})\equiv \chi_{B_R}(x).
			\end{align}
			Evaluating \eqref{3122'} at the initial state yields
			\begin{align}\label{local-est4}
				\big \lan N_{f,0s} \big\ran_0 \le\big \langle N_{B_R}\big\rangle_0.\end{align}

			\begin{figure}[H]
				\centering
				\begin{tikzpicture}[scale=3]
					\draw [->] (-.5,0)--(2,0);
					\node [right] at (2,0) {$\mu$};
					\node [below] at (.3,0) {$0$};
					\draw [fill] (.3,0) circle [radius=0.02];

					\draw [very thick] (-.5,0)--(.3,0);
					\draw [very thick] (.3,1)--(2,1);				
					\filldraw [fill=white] (.3,1) circle [radius=0.02];
					
					\draw [dashed, very thick] (-.5,0)--(.75,0) [out=20, in=-160] to (1.5,1)--(2,1);

					\draw [->] (1.55,.5)--(1.3,.5);
					\node [right] at (1.55,.5) {$f(\tfrac\mu s)$};
					
					\draw [->] (0,.5)--(.85,.95);
					\draw [->] (0,.5)--(-.05,.05);
					\node [left] at (0,.5) {$\theta(\mu)$};
					
				\end{tikzpicture}
				\caption{Schematic diagram illustrating \eqref{3122}}
				\label{fig:f0}
			\end{figure}
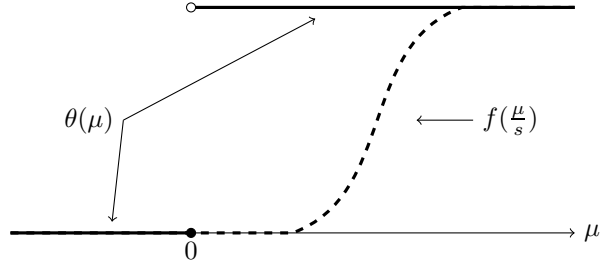

			2. Next, retaining the first term and dropping the second one in the l.h.s.~of \eqref{propag-est1}, and then using \eqref{local-est4} to bound the first term on the r.h.s., we find that there exist $C=C(f, \al,d, C_J)>0$ and functions $j_k\in\cE$ s.th.~for all $s,\,t>0$,
			\begin{align} \label{propEst10}
				&\big\lan N_{f,ts} \big \ran_t \le \,\big \langle N_{B_R}\big\rangle_0 +C \sum_{k=2}^n s^{-k}\int_0^t \big\lan N_{j_k',\tau s}\big \ran_\tau d\tau + C s^{-n-1}t\Rem(t). 
			\end{align}
			Applying \eqref{intRME} and again \eqref{local-est4} to estimate the integrated term, we deduce that for some $C=C(f, \al,d, C_J,v)>0$ and all $s\ge t>0$,
			\begin{align} \label{propag-est4} 
				\big\lan N_{f,ts} \big\ran_t \le(1+Cs^{-1}) \big\langle N_{B_R} \big\rangle_0 +  C s^{-n}  \Rem(s).  \end{align}

			3.  Again by the definition of $ \cE $,  setting \begin{align}
				\label{deltaChoice}
				\eps=v-v'>0,
			\end{align} we have $f(\frac{\mu -v't}{s})\equiv 1$ for all $\mu\ge v't+(v-v')s$. 
			
			Now, we choose 
			\begin{align}
				\label{sChoice}
				s=\eta/v>0,\quad \eta=R-r.
			\end{align} Then, for all $t< s$ and $v'<v$, we have $f(\frac{\mu-v't}{s})\equiv 1$ for  $\mu\ge \eta$. 
			This implies the estimate \begin{equation}\label{3132}
				f((\mu-v't)/s)\ge \theta(\mu-\eta),
			\end{equation}see \figref{fig:f}.
			For $\mu=R-\abs{x}$, \eqref{3132} implies
			\begin{align}
				\label{3132'}
				f_{ts}(x)\equiv f\del{\frac{R-\abs{x}-v't}{s}}\ge\theta(b-\abs{x})\equiv \chi_{B_b}(x).
			\end{align}
			Evaluating \eqref{3132'} at $\psi_t$ yields
			\begin{align}\label{locEstT}
				\big\lan N_{B_r} \big\ran_t &\le \big\lan N_{f,ts}\big \ran_t .
			\end{align}
			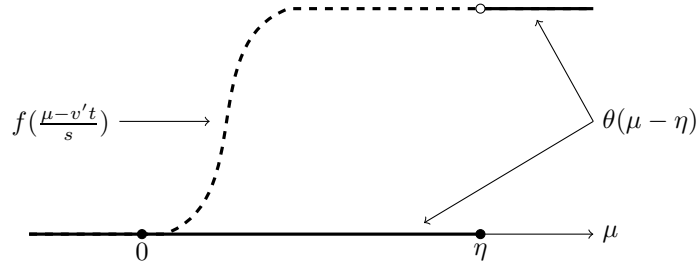
\begin{figure}[H]
				\centering
				\begin{tikzpicture}[scale=3]
					\draw [->] (-.5,0)--(2,0);
					\node [right] at (2,0) {$\mu$};
					
					\node [below] at (0,0) {$0$};
					\draw [fill] (0,0) circle [radius=0.02];
					
					\node [below] at (1.5,0) {$\eta$};
					\draw [fill] (1.5,0) circle [radius=0.02];

					\draw [very thick] (-.5,0)--(1.5,0);
					\draw [very thick] (1.5,1)--(2,1);
					\draw [dashed, very thick] (-.5,0)--(.1,0) [out=20, in=-160] to (.65,1)--(2,1);
					\filldraw [fill=white] (1.5,1) circle [radius=0.02];

					\draw [->] (-.1,.5)--(.3,.5);
					\node [left] at (-.1,.5) {$f(\tfrac{\mu-v't}{s})$};
					
					\draw [->] (2,.5)--(1.75,.95);
					\draw [->] (2,.5)--(1.25,.05);
					\node [right] at (2,.5) {$\theta(\mu-\eta)$};
					
				\end{tikzpicture}
				\caption{Schematic diagram illustrating \eqref{3132}.}
				\label{fig:f}
			\end{figure}
			
			4. Lastly, applying \eqref{locEstT} to estimate the l.h.s. of \eqref{propag-est4}, we conclude that for some ${C=C(f, \al,d, C_J,v)>0}$ and all $s\ge t>0$,
			\begin{align}
				\label{MVBout}
				\big\lan N_{B_r} \big\ran_t \le(1+Cs^{-1}) \big\langle N_{B_R} \big\rangle_0 +  C s^{-n}  \Rem(s),
			\end{align}
			which gives the desired estimate \eqref{propEst1}.
			
			\qed

			
			\section{Proof of \thmref{pMVB main}}\label{sec proof multi}
			In this section, we complete the proof of \thmref{pMVB main} by extending Theorem \ref{MVBmain} to higher moments. 
			
			First, we state the following higher order integral estimate:
			\begin{proposition}\label{propIntRMEp} 
				Let \eqref{Kcond} hold for $n\ge1$. 
				Then, for any $v>\kappa$,  $f\in\mathcal{E}$, and integer $p\ge1$, there exist $C=C(f,n,v,p)>0$ and functions $j_k\in\mathcal{E}$, $1\leq k\leq n$, such that for all $t,s,R>0$,
				\begin{align}\label{intRMEp}
					&\int_0^t \langle (N_{f,{\tau}s}+1)^{p-1}N_{f', {\tau}s }\rangle_{\tau}d{\tau}\le C\del{  \sum_{q=1}^p\sum_{k=1}^ns^{-k+2}\langle N_{j_k,0s}^{q} \rangle_0+   s^{-n}t\TRem_p(t)},\\
					&\TRem_p(t):=\sup_{{\tau}\le t}{\sum_{q=1}^p\del{\langle (N_{f,{\tau}s}+1)^{q-1}N_{B_{R-\eps s/2}}\rangle_{{\tau}}+\sum_{x\in B_{R-\eps s/2}}\sum_{y\in\Lambda}\abs{J_{xy}}|x-y|^{n+1}\langle (N_{f,{\tau}s}+1)^{q-1}n_y\rangle_{\tau} }}.\label{TRemDef}
				\end{align} 
			\end{proposition}
			\propref{propIntRMEp} leads to the following bound on the evolution of the $p$-th power of the number operators.
			\begin{proposition}\label{propMVBp} Let \eqref{Kcond} hold for $n\ge1$. 
				Then, for any   $v>\kappa$ and integer $p\ge1$, there exists $C=C(n, C_J, v,p)>0$ such that for all $R>r>0$ and $s=(R-r)/v$,
						\begin{align}\label{propEstp}
						\sup_{t\le s}	\br{N_{B_r}^p}_t\le& {\left(1 +Cs^{-1} \right){\br{N_{B_R}^p }_0}} + Cs^{-n}{\HRem_p(s)},\\
						\HRem_p(s):=&\sup_{t\le s}\del{\langle N_{B_{R}}^p\rangle_{t}+\sum_{x\in B_{R}}\sum_{y\in\Lambda}\abs{J_{xy}}|x-y|^{n+1}\langle (1+N_{B_R}^{p-1})n_y\rangle_t }.\label{HRemDef}
					\end{align}
				
			\end{proposition}
			The proofs of Props.~\ref{propIntRMEp}--\ref{propMVBp}  are deferred to Sects.~\ref{secPf51}--\ref{secPf52}. 
			
			\propref{propMVBp} allows us to prove the following result, which leads to \thmref{pMVB main}  by rescaling, as in the proof of  \thmref{MVBmain}.
			
			\begin{proposition}\label{propMultScalep}

				Let {$p\ge2$} be an integer and let \eqref{Jcond} hold with {$\al>3dp/2+1$}. 
				Then, for any $v > \kappa$, there exists $C=C(\al,d, C_J, v,p)>0$ such that for all $\l,\,k\ge0$, and initial states satisfying \eqref{UDBp},
				there holds
				{		\begin{align}
						\label{multEstp}
						\sup_{0\le vt\le 2^k}\br{N_{B_{2^k}}^p}_t\le {\left(1+2^{-k}C\right){\br{N_{B_{2^{k+1}}}^p }_0}}+ C\l^p .
				\end{align}}
			\end{proposition}

			\begin{proof}
				We follow the downscale induction scheme as in the proof of \propref{propMultScale}. 
				Within this proof we set		 
				\begin{align}
					\label{nDefp}
					n:=\lfloor\al-{\frac{dp}{2}}-1\rfloor.
				\end{align}

				1.	By definition \eqref{HRemDef}, condition \eqref{Kcond},  and the conservation of $N$, we find the crude remainder estimate $\HRem_p(t)\le \br{N^p}_0$ for all $t$. This fact, together with estimate \eqref{propEstp} and the uniform density assumption \eqref{UDBp}, implies that \eqref{multEstp} holds for sufficient large $K$ and all $k\ge K$. (See Step 1 in \propref{propMultScale} for details.) This establishes the base case. Now, assuming \eqref{multEstp} holds for $k+1$, we prove it for $k$.

				2. Recalling definition \eqref{HRemDef} for $\HRem_p$, we apply \propref{propMVBp} with $R=2^{k+1}$ and $r=2^k$, whence $s=2^k/c$ and
				\begin{align}
					\sup_{vt\le 2^k}	\br{N_{B_{2^k}}^p}_t
					\le&{\left(1+2^{-k}C\right){\br{N_{B_{2^{k+1}}}^p }_0}}\label{propEst2p}\\
					&+2^{-nk}C\sup_{vt\le 2^k}\langle N_{B_{2^{k+1}}}^p\rangle_{t}\label{propEst2p'}\\
					&+2^{-nk}C\sup_{vt\le 2^k}\sum_{x\in B_{2^{k+1}}}\sum_{y\in\Lambda}\abs{J_{xy}}|x-y|^{n+1}\langle (1+N_{B_{2^{k+1}}}^{p-1})n_y\rangle_{t}.\label{propEst2p''}
				\end{align}
				The terms in the r.h.s.~of line \eqref{propEst2p} {are of the desired form}. To handle second term in line \eqref{propEst2p'}, we use
				the induction hypothesis and the uniform density assumption \eqref{UDBp} to obtain 
				\begin{align}
					\label{propEst13}
					2^{-dpk}	\sup_{vt\le 2^k}\langle N_{B_{2^{k+1}}}^p\rangle_{t}
					\le& 2^{-dpk}C{\del{\del{1+2^{-(k+1)}}{\br{N_{B_{2^{k+2}}}^p}_0}+\l^p}}\notag\\
					\le&C\del{\del{1+2^{-(k+1)}} 2^{2dp}+{2^{-dpk}}}\l^p .
				\end{align}
				We conclude from \eqref{propEst13} that
				\begin{align}
					\label{propEst14}
					2^{-dpk}\sup_{vt\le 2^k}\langle N_{B_{2^{k+1}}}^p\rangle_{t}\le C \l^p .
				\end{align}
				Since 	$n\ge dp$ by the assumption {$\al>3dp/2+1$} (see \eqref{nDefp}), applying estimate \eqref{propEst14}  yields the desired upper bound for the term in line \eqref{propEst2p'}.

				3.	It remains to bound the term appearing in line \eqref{propEst2p''}. We write
				\begin{align}
					&2^{-nk}\sup_{vt\le 2^k}\sum_{x\in B_{2^{k+1}}}\sum_{y\in\Lambda}\abs{J_{xy}}|x-y|^{n+1}\langle (1+N_{B_{2^{k+1}}}^{p-1})n_y\rangle_{t}\le \Mr{I}+\Mr{II}+\mathrm{I}_p+\mathrm{II}_p,\label{12p}
				\end{align}
				where $\Mr{I}$, $\Mr{II}$ are respectively given by \eqref{Idef}, \eqref{IIdef}, and 
				\begin{align}
					\mathrm{I}_p:=&2^{-nk} \sup_{vt\le 2^k}\sum_{\abs{x}\le 2^{k+1}}\sum_{\abs{y}\le 2^{k+1}}\abs{J_{xy}}|x-y|^{n+1}\langle  N_{B_{2^{k+1}}}^{p-1} n_y\rangle_{t},\label{Ipdef}\\ 
					\mathrm{II}_p:=&2^{-nk} \sup_{vt\le 2^k} \sum_{\abs{x}\le 2^{k+1}} \sum_{l> k+1}   \sum_{2^l<\abs{y}\le 2^{l+1}} \abs{J_{xy}}|x-y|^{n+1}\langle N_{B_{2^{k+1}}}^{p-1}n_y\rangle_{t}\label{IIpdef}.
				\end{align}
			Below we bound the four terms in the r.h.s. of \eqref{12p}.
			
		3.1. 	Since $N_X,\,X\subset\Lam$ has integer eigenvalues, we have 
			\begin{align}\label{Nrel}
				N_X^p\ge N_X,\quad X\subset \Lam.
			\end{align}
		This, together with the uniform density  condition \eqref{UDBp}, implies $\br{n_x}_0\le \min(\l,\l^p)$ for any $x\in\Lam$, and so
		\begin{align}\label{NXest1}
			\br{N_X}_0 \le \min (\l,\l^p)\abs{X},\quad X\subset\Lam.
		\end{align}
	By \eqref{NXest1}, estimates \eqref{Iest2'} and \eqref{IIest2} hold with $\l$ replaced by $\min (\l,\l^p)$ under condition \eqref{UDBp}. Thus  we conclude
	\begin{align}\label{12Est}
		\Mr{I}+\Mr{II}\le C\l^p.
	\end{align}

		3.2. Next, we bound $\Mr{I}_p$. 	By condition \eqref{Kcond}, the term $\Mr{I}_p$ in \eqref{Ipdef} can be bounded as 
				\begin{align}\label{Iest1p}
					\Mr{I}_p\le&2^{-nk}\del{\sup_{x\in\Lam}\sum_{y\in\Lam} \abs{J_{xy}}\abs{x-y}^{n+1}}
					\del{\sup_{vt\le 2^k}\sum_{\abs{y}\le 2^{k+1}}\langle N_{B_{2^{k+1}}}^{p-1} n_y\rangle_{t}}\notag
					\\\le & 2^{-nk}C\sup_{vt\le 2^k}\br{ N_{B_{2^{k+1}}}^p}_t.
				\end{align}
				Applying \eqref{propEst14} to \eqref{Iest1p}, we find that for $n\ge dp $, 
				\begin{align}\label{Iest2'p}
					\Mr{I}_p	\le&C  \l^p  .
				\end{align}

		3.3.	Lastly, we bound $\Mr{II}_p$. For $\g:=\al-n-1$ (see \eqref{gDef}),   we bound  the second term $\Mr{II}_p$ in \eqref{IIpdef} {using the growth condition \eqref{growthCond2}} as,  
			\begin{align}
				\Mr{II}_p\le& 2^{-nk}C_{d,p}\sup_{vt\le 2^k}\sum_{\abs{x}\le 2^{k+1}} \sum_{l> k+1} \sum_{2^l<\abs{y}\le 2^{l+1}}   \abs{x-y}^{-\g}\langle  N_{B_{2^{k+1}}}^{p-1} n_y\rangle_{t}\notag\\
				\le& 2^{-nk} C_{d,p} \sup_{vt\le 2^k}\sum_{\abs{x}\le 2^{k+1}} \sum_{l> k+1}\sum_{2^l<\abs{y}\le 2^{l+1}} \del{2^l-\abs{x}}^{-\g} \langle N_{B_{2^{k+1}}}^{p-1}n_y\rangle_{t}
				\notag\\
\le& 2^{-nk} C_{d,p} \sup_{vt\le 2^k}\sum_{\abs{x}\le 2^{k+1}} \sum_{l> k+1} \del{2^l-\abs{x}}^{-\g} \langle N_{B_{2^{k+1}}}^{p-1}N_{B_{2^{l+1}}}\rangle_{t}
				\notag\\
				\le & 2^{-nk}C_{d,p}\sum_{m=0}^{2^{k+1}}m^{d-1} \sum_{l> k+1}\del{2^l-m}^{-\g}\sup_{vt\le 2^k}\br{{N_{B_{2^{k+1}}}^{p/2}N_{B_{2^{l+1}}}^{p/2}}}_t\notag\\
				\le & 2^{-nk}C_{d,p}\sum_{m=0}^{2^{k+1}}m^{d-1} \sum_{l> k+1}\del{2^l-m}^{-\g}\sup_{vt\le 2^k}{\br{N_{B_{2^{k+1}}}^{p}}_t^{1/2}\br{N_{B_{2^{l+1}}}^{p}}_t^{1/2}}\notag\\
				\le & C_{d,p}\underbrace{\roun{2^{-(n-d)k}\sup_{vt\le 2^k}\br{N_{B_{2^{k+1}}}^{p}}_t}^{1/2}}_{\Mr{A}}\underbrace{\del{2^{-dk}\sum_{m=0}^{2^{k+1}}m^{d-1}} }_{\Mr{B}}\underbrace{\del{\sum_{l> k+1}2^{-(l-1)\g}{\roun{\sup_{vt\le 2^k}\br{N_{B_{2^{l+1}}}^p}_t}^{1/2}}}}_{\Mr{C}}
				.\label{IIest1p'}
			\end{align}
			
			{The term A in line \eqref{IIest1p'} can be bounded using \eqref{propEst14} as
				\begin{align}\label{522}
					\Mr{A}\le C 2^{-(n-d)k} 2^{dpk/2}\l^{p/2}.
				\end{align}
			}
		{Since $n\ge dp$ (see \eqref{nDefp}), we have $n-d\ge d(p-1)$. This relation, together with \eqref{522}, implies 
			\begin{align}
				\label{Abdd}
				\Mr{A} \le C \l^{p/2} \iff p\ge2.
			\end{align}
		Notice that at this point we use the assumption $p\ge2$ in a crucial way. 
		}
			The term B can be bounded identically as in \eqref{IIest11}.
			To bound the term C, we apply once again estimate \eqref{propEst14} as in \eqref{522}, yielding
			\begin{align}
				\label{IIest12p}
				\Mr{C}\le
				\sum_{l> k+1}2^{-(l-1)\g+dpl/2} 
					C \l^{p/2}.
			\end{align}
			By the definitions of $n,\,\g$ in \eqref{nDefp}, \eqref{gDef}, we have {$\g>dp/2$} and so $\sum_{l> k+1}2^{-(l-1)\g+{dpl/2}} $ converges geometrically. Thus we obtain
			\begin{align}
				\label{Cbdd}
				\Mr{C}\le C \l^{p/2}.
			\end{align}
Combining \eqref{Abdd}, \eqref{IIest11}, and \eqref{Cbdd},   we conclude 
			\begin{align}
				\label{IIest2p}
				\Mr{II}_p\le C\l^p.
			\end{align}
			
			3.4. Inserting \eqref{12Est}, \eqref{Iest2'p}, and \eqref{IIest2p} into \eqref{12p}  yields the desired upper bound for the terms in line \eqref{propEst2p''}. This completes the induction step and \propref{propMultScalep} is proved.
		\end{proof}
		
		{
			\begin{proof}[Proof of \thmref{pMVB main}]
				The proof can be easily adapted from that of \propref{propMultScalep}. See Section \ref{sec proof MVB} for details. 
			\end{proof}

		}

		\subsection{Proof of \propref{propIntRMEp}}\label{secPf51}
		We use an induction in $p$, with the base case $p=1$  established in \propref{lemLocIntRME}.
		
		1. To begin with, for fixed $p\ge1,\,s>0,\,f\in\cE$, we define, in place of \eqref{ASTLO1}, the higher order ASTLOs
		\begin{align}\label{PhiDefp}
			\Phi(t)=(N_{f,ts})^p.
		\end{align}
		From the Leibniz rule, we find
		\begin{align}
			D\Phi(t)
			= &\frac{-pv'}{s}N_{f,ts}^{p-1}N_{f',ts}
			+[iH,\Phi(t)] \notag\\
			=&\frac{-pv'}{s}N_{f,ts}^{p-1}N_{f',ts}
			+\frac{1}{2}\sum_{q=0}^{p-1} \left(N_{f,ts}^q[iH,N_{f,ts}]N_{f,ts}^{p-q-1}+\mathrm{h.c.}\right) \label{basicEqP}.
		\end{align}
		(Recall that $\mathrm{h.c.}$ stands for Hermitian conjugate.) In the second line, we used that $D\Phi$ is Hermitian to make each summand manifestly Hermitian.

		2. Our next goal is to apply Cauchy-Schwarz inequality to bound the second term in line \eqref{basicEqP}. For this, we move the number operators in between  $a_x^*$ and $ a_y$ by using the commutator identities
		\begin{align}
			\label{commutator}
			&[N_{f,ts},a_x^*]=\sum_z f_{ts}(|z|) [n_z,a_x^*]
			=\sum_z f_{ts}(|z|) a_z^* \delta_{z,x}
			=f_{ts}(|x|) a_x^*,\\
			&[a_y, N_{f,ts},]=\sum_z f_{ts}(|z|) [a_y,n_z]
			=\sum_z f_{ts}(|z|) \delta_{y,z}a_z
			=f_{ts}(|y|) a_y.\label{commutator'}
		\end{align}
		A straightforward induction based on \eqref{commutator}--\eqref{commutator'} gives
		\begin{align}\label{NaRel}
			N_{f,ts}^q a_x^* =a_x^* (N_{f,ts}+f_{ts}(|x|))^q,\qquad a_y N_{f,ts}^{p-q-1}=(N_{f,ts}+f_{ts}(|y|))^{p-q-1} a_y.
		\end{align}
		
		For simplicity notation, below we set
		\begin{align}
			A(z):=N_{f,ts}+f_{ts}(|z|).\label{Adef}
		\end{align}
		By Lemma \ref{lem:HRf-com} and relation \eqref{NaRel}, the commutator term in line \eqref{basicEqP} can be written as
		\begin{align}
			[iH,\Phi(t)] 
			=\frac{1}{2}\sum_{x\neq y} J_{xy}(f_{ts}(|x|)-f_{ts}(|y|))\sum_{q=0}^{p-1} \left(ia_x^* A(x)^qA(y)^{p-q-1}a_y+\mathrm{h.c.}\right).\label{HNpcomm}
		\end{align}
		Applying estimate \eqref{fDifEst2} to bound the difference term $f_{ts}(|x|)-f_{ts}(|y|)$ in \eqref{HNpcomm}, we find
		\begin{align}
			&\quad[iH,\Phi(t)]  \notag\\&\le  \sum_{x\neq y} \abs{J_{xy}} \frac{\abs{x-y}}{s}u(|x|_{ts})u(|y|_{ts}) \sum_{q=0}^{p-1}\left(ia_x^* A(x)^qA(y)^{p-q-1}a_y\right)\notag \\
			&\quad + \sum_{k=2}^n \sum_{x\neq y} \abs{J_{xy}} \frac{\abs{x-y}^k}{s^k} \abs{h_k(|x|_{ts},|y|_{ts})} \sum_{q=0}^{p-1}\left(ia_x^* A(x)^qA(y)^{p-q-1}a_y\right)\notag \\
			&\quad + \frac12\sum_{x\neq y} \abs{J_{xy}} \frac{\abs{x-y}^{n+1}}{s^{n+1}}\abs{R_n(|x|_{ts},|y|_{ts})}\sum_{q=0}^{p-1}\left(ia_x^* A(x)^qA(y)^{p-q-1}a_y\right)\del{\chi_{B_{R-\eps s/2}}(x)+\chi_{B_{R-\eps s/2}}(y)}\notag \\
			&\quad +\hc \label{HNpEst1}
		\end{align}
		
		Recall that $|h_k(x,y)|\leq \tilde u_k(x)\tilde u_k(y)$ and ${\abs{R_n(x,y)}\le C_{f,n}}$ (see \eqref{hkEst}--\eqref{RnEst}). Applying these and the triangle inequality to \eqref{HNpEst1}, we find
		\begin{align}
			&\quad[iH,\Phi(t)]  \notag\\&\le  \sum_{x\neq y} \abs{J_{xy}} \frac{\abs{x-y}}{s} \sum_{q=0}^{p-1}\left(iu(|x|_{ts})a_x^* A(x)^qA(y)^{p-q-1}a_yu(|y|_{ts})\right)\notag \\
			&\quad + \sum_{k=2}^n \sum_{x\neq y} \abs{J_{xy}} \frac{\abs{x-y}^k}{s^k} \sum_{q=0}^{p-1}\left(i\tilde u_k(x)a_x^* A(x)^qA(y)^{p-q-1}a_y\tilde u_k(y)\right)\notag \\
			&\quad + C_{f,n}\sum_{x\neq y} \abs{J_{xy}} \frac{\abs{x-y}^{n+1}}{s^{n+1}}\sum_{q=0}^{p-1}\left(ia_x^* A(x)^qA(y)^{p-q-1}a_y\right)\del{\chi_{B_{R-\eps s/2}}(x)+\chi_{B_{R-\eps s/2}}(y)}\notag \\
			&\quad +\hc \label{HNpEst2}
		\end{align}


		3. Now that we have moved $A(x)^qA(y)^{p-q-1}$ to the middle, we can apply Cauchy-Schwarz for operators to split up the $a_x^*$ and $a_y$ and estimate them by number operators.  
		For this we note that $A(z)$ are positive definite and $[A(x),A(y)]=0$ for any $x,\,y\in\Lam$. This implies that $A(x)^qA(y)^{p-q-1}$ is positive definite and has a well-defined operator square root $\sqrt{A(x)^qA(y)^{p-q-1}}$. At this point, it suffices to apply Cauchy-Schwarz as follows: for any function $g:\Lam\to\Rb$,
		\begin{align}
			\label{CSg}
			& g(x)a_x^* {A(x)^qA(y)^{p-q-1}}a_y g(y)+\mathrm{h.c.}\notag\\
			=&( g(x)a_x^* \sqrt{A(x)^qA(y)^{p-q-1}} \sqrt{A(x)^qA(y)^{p-q-1}}a_y g(y)+\mathrm{h.c.})\notag\\
			\leq &g(x)^2 a_x^* A(x)^qA(y)^{p-q-1} a_x +a_y^* A(x)^qA(y)^{p-q-1} a_yg(y)^2.
		\end{align}
		At this stage, we can estimate $f\leq 1$ inside \eqref{Adef} and so, for any $q=0,\ldots,p-1$,
		\begin{align}\label{CSg2}
			&g(x)^2 a_x^* A(x)^qA(y)^{p-q-1} a_x +a_y^* A(x)^qA(y)^{p-q-1} a_yg(y)^2\notag\\\leq &
			g(x)^2 a_x^* (N_{f,ts}+1)^{p-1} a_x +a_y^* (N_{f,ts}+1)^{p-1} a_yg(y)^2.
		\end{align}
		It remains to merge again the $a_x^*$ and $a_x$ that are separated by the $(N_{f,ts}+1)^{p-1} $ to create $n_x$ and therefore the desired ASTLO. 
		For this, we undo the commutation procedure. By another simple induction based on relations \eqref{commutator}--\eqref{commutator'}, we obtain
		\begin{align}\label{commRel3}
			a_x^*(N_{f,ts}+1)^{p-1}a_x
			=n_x(N_{f,ts}+1-f_{ts}(x))^{p-1}
			\leq n_x(N_{f,ts}+1)^{p-1},
		\end{align}
		where the inequality holds since $f\geq 0$. Inserting \eqref{commRel3} into \eqref{CSg2} yields, for any $q=0,\ldots,p-1$,
		\begin{align}\label{CSg3}
			&	g(x)^2 a_x^* A(x)^qA(y)^{p-q-1} a_x +a_y^* A(x)^qA(y)^{p-q-1} a_yg(y)^2\notag\\\leq &
			(g(x)^2 n_x 
			+ g(y)^2 n_y)(N_{f,ts}+1)^{p-1}.
		\end{align}
		
		4. Combining \eqref{HNpEst2} with estimates \eqref{CSg}, \eqref{CSg2}, and \eqref{CSg3}, and recalling $u^2=f'$, $\tilde u_k^2=j_k'$ (see \lemref{lemSymExp}), we conclude that
		\begin{align}
			&\quad[iH,\Phi(t)]  \notag\\&\le  p(N_{f,ts}+1)^{p-1}\sum_{x\neq y} \abs{J_{xy}} \frac{\abs{x-y}}{s} \left(f'(x) n_x 
			+ f'(y) n_y\right)\notag \\
			&\quad + p(N_{f,ts}+1)^{p-1}\sum_{k=2}^n \sum_{x\neq y} \abs{J_{xy}} \frac{\abs{x-y}^k}{s^k} \left(j_k'(x) n_x 
			+ j_k'(y) n_y\right)\notag \\
			&\quad + pC_{f,n}(N_{f,ts}+1)^{p-1}\sum_{x\neq y} \abs{J_{xy}} \frac{\abs{x-y}^{n+1}}{s^{n+1}}\left(n_x 
			+ n_y\right)\del{\chi_{B_{R-\eps s/2}}(x)+\chi_{B_{R-\eps s/2}}(y)}. \label{HNpEst3}
		\end{align}
		We proceed to bound the contributions of the three terms in the r.h.s.~of \eqref{HNpEst3} respectively as in \eqref{1stEst}, \eqref{2ndEst}, and \eqref{R1Est7}. The details are analogous to Steps 2-3 in the proof of estimate \propref{lemLocIntRME} and therefore omitted. Here we record the result (c.f. \eqref{ME1}):
		\begin{align}
			\label{HNpEst4}
			{i\big [H , \Phi(t) \big ]}\le& p(N_{f,ts}+1)^{p-1}\kappa s^{-1} N_{f',ts} +  pC_{f,n}(N_{f,ts}+1)^{p-1} \sum_{k=2}^n s^{-k} N_{j_k' ,ts}   \notag\\&+ \frac{ pC_{f,n}(N_{f,ts}+1)^{p-1}}{s^{n+1}} \left( N_{B_{R-\eps s/2}}+\sum_{x\in B_{R-\eps s/2}}\sum_{y\in\Lambda}\abs{J_{xy}}|x-y|^{n+1}{ n_y} \right).
		\end{align}

		Going back to relation \eqref{basicEqP} and isolating the first order terms, we find (c.f. \eqref{ME2})
		\begin{align}\label{DPhiEst1}
			D\Phi(t)\le& {pv'  s^{-1}((N_{f,ts}+1)^{p-1}-N_{f,ts}^{p-1})N_{f',ts}}\\
			&+{p(\kappa-v') s^{-1}(N_{f,ts}+1)^{p-1} N_{f',ts}}
			\notag\\&+pC_{f,n}(N_{f,ts}+1)^{p-1}\sum_{k=2}^n s^{-k} N_{j_k' ,ts}\notag\\&+\frac{ pC_{f,n}(N_{f,ts}+1)^{p-1}}{s^{n+1}} \left( N_{B_{R-\eps s/2}}+\sum_{x\in B_{R-\eps s/2}}\sum_{y\in\Lambda}\abs{J_{xy}}|x-y|^{n+1}{ n_y} \right)\notag.
		\end{align}
		
		5. The last three terms are of the same form as in \eqref{ME2} and can be handled as in the proof of Prop.~\ref{lemLocIntRME}. The first term appearing in line \eqref{DPhiEst1} is the price we paid out for commutators. While it is of order $s^{-1}$, the key point is that each commutator removed one $N_{f,ts}$ and accordingly also this term has one less order.
		More precisely, we claim that for any $p\ge1$, there exist $C=C(f,v,n,p)>0$ and functions $\tilde f \ge f $ and, for $n\ge2$, some  $\tilde j_2,\ldots, \tilde j_n$ in $\cE$ with $\tilde j_k'\ge j_k'$, s.th.
		\begin{align}
			\label{intpEst2}&\int_0^t \langle (N_{f,{\tau}s}+1)^{p-1} N_{f',{\tau}s}\rangle_{\tau} 
			\notag\\\leq& C\del{s \sum_{q=1}^p\langle N_{\tilde f,0s}^{q} \rangle_0+ \sum_{q=1}^p\sum_{k=2}^n s^{-k+1} \int_0^t \big\lan (N_{ f,ts}+1)^{q-1}N_{\tilde j_k',{\tau}s} \big\ran_{\tau} d{\tau}+  s^{-n}\int_0^t\tilde Q_p({\tau})},\\
			&\tilde Q_p({\tau}):=\sum_{q=1}^p\del{\langle (N_{f,ts}+1)^{p-1}N_{B_{R-\eps s/2}}\rangle_{{\tau}}+\sum_{x\in B_{R-\eps s/2}}\sum_{y\in\Lambda}\abs{J_{xy}}|x-y|^{n+1}\langle (N_{f,ts}+1)^{p-1}\label{RemTilDef} n_y\rangle_{{\tau}} }.
		\end{align}
		We prove \eqref{intpEst2} by induction. For the base case $p=1$, \eqref{intpEst2} is exactly the same as \eqref{propag-est3}. Assuming now \eqref{intpEst2} holds for $p-1$, we prove it for $p$. 
		
		Using relation \eqref{eq-basic} and recalling definition \eqref{PhiDefp}, we take expectation of and integrate estimate \eqref{DPhiEst1} to obtain
		\begin{align}
			\label{DPhiEst2}
			&\br{N_{f,ts}^p}_t+p(v'-\kappa)s^{-1}\int_0^t \langle (N_{f,{\tau}s}+1)^{p-1} N_{f',{\tau}s}\rangle_{\tau}\notag\\ \le&\br{N_{f,0s}^p}_0+ pv'  s^{-1}\int_0^t\br{(N_{f,{\tau}s}+1)^{p-1}-N_{f,{\tau}s}^{p-1})N_{f',{\tau}s}}_{\tau}\,d{\tau}\\
			& +pC_{f,n}\del{\sum_{k=2}^n s^{-k} \int_0^t \big\lan (N_{f,{\tau}s}+1)^{p-1}N_{j_k',{\tau}s} \big\ran_{\tau} d{\tau}+  s^{-n-1}\int_0^tQ_p({\tau})\,d{\tau}},\notag
		\end{align}
		where we set
		\begin{align}\label{QpDef}
			Q_p({\tau}):=&\del{\langle (N_{f,{\tau}s}+1)^{p-1}N_{B_{R-\eps s/2}}\rangle_{{\tau}}+\sum_{x\in B_{R-\eps s/2}}\sum_{y\in\Lambda}\abs{J_{xy}}|x-y|^{n+1}\langle (N_{f,{\tau}s}+1)^{p-1}\ n_y\rangle_{{\tau}} }.
		\end{align}
		Estimate \eqref{DPhiEst2} plays the role of \eqref{propag-est1}, except for that now we have to deal with the new difference term appearing in the r.h.s..
		
		\begin{align}
			\label{MVT}
			(h+1)^{p-1}-h^{p-1}\leq (p-1)(h+1)^{p-2},
		\end{align}
		and so
		\[
		(N_{f,ts}+1)^{p-1}-N_{f,ts}^{p-1}\leq
		(p-1)(N_{f,ts}+1)^{p-2}.
		\]
		Thus \eqref{DPhiEst2} becomes
		\begin{align}
			\label{DPhiEst3}
			&\br{N_{f,ts}^p}_t+p(v'-\kappa)s^{-1}\int_0^t \langle (N_{f,{\tau}s}+1)^{p-1} N_{f',{\tau}s}\rangle_{\tau}\notag\\ \le&\br{N_{f,0s}^p}_0+ {p(p-1)v'  s^{-1}\int_0^t\br{(N_{f,{\tau}s}+1)^{p-2}N_{f',{\tau}s}}_{\tau}\,d{\tau}}\\
			& +pC_{f,n}\del{\sum_{k=2}^n s^{-k} \int_0^t \big\lan (N_{f,{\tau}s}+1)^{p-1}N_{j_k',{\tau}s} \big\ran_{\tau} d{\tau}+  s^{-n-1}\int_0^tQ_p({\tau})\,d{\tau}}.\notag
		\end{align}
		Rearranging estimate \eqref{DPhiEst3} as in {the derivation of \eqref{propag-est3}}, we  obtain:
		\begin{align}
			&\int_0^t \langle (N_{f,{\tau}s}+1)^{p-1} N_{f',{\tau}s}\rangle_{\tau} 
			\leq (p-1)C_{f,c,n} \int_0^t \langle (N_{f,{\tau}s}+1)^{p-2} N_{f',{\tau}s}\rangle_{\tau}\label{intpEst1}
			\\&\qquad+C_{f,c,n}\del{{p^{-1}s} \langle N_{f,0s}^{p} \rangle_0+ \sum_{k=2}^n s^{-k+1} \int_0^t \big\lan (N_{f,{\tau}s}+1)^{p-1}N_{j_k',{\tau}s} \big\ran_{\tau} d{\tau}+  s^{-n}\int_0^tQ_p({\tau})\,d{\tau}}.\notag
		\end{align}

		Note that the term appearing in the r.h.s.~of line \eqref{intpEst1} is of the same form as in the l.h.s., but with exactly one order less in $p$. Hence, applying the induction hypothesis and using property \eqref{E2} {(see \secref{secE})}, we conclude \eqref{intpEst2}. This completes the induction and \eqref{intpEst1} is proved.
		
		Using estimate \eqref{intpEst1}, H\"older's inequality, and the basic properties of class $\cE$ as above, we can derive iteratively estimate \eqref{intRMEp} as in Step 6 in the proof of Prop.~\ref{lemLocIntRME}. This completes the proof of \propref{propIntRMEp}.  
		\qed
		
		\subsection{Proof of \propref{propMVBp}}\label{secPf52}
		We follow the proof of \propref{propMVB1} and omit detailed derivation that can be easily adapted from the latter. 
		
		Fix any function $f\in\cE$. Rearranging estimate \eqref{DPhiEst3} as in {the derivation of \eqref{propEst10}} and recalling definitions \eqref{QpDef}, \eqref{TRemDef}, we  obtain:
		\begin{align}
			\label{propEst11}
			\br{N_{f,ts}^p}_t\le&\br{N_{f,0s}^p}_0+ {p(p-1)v'  s^{-1}\int_0^t\br{(N_{f,{\tau}s}+1)^{p-2}N_{f',{\tau}s}}_{\tau}\,d{\tau}}\\
			& +pC_{f,n}\del{\sum_{k=2}^n s^{-k} \int_0^t \big\lan (N_{f,{\tau}s}+1)^{p-1}N_{j_k',{\tau}s} \big\ran_{\tau} d{\tau}+  s^{-n-1}t\TRem_p(t)}.\notag
		\end{align}
		Note that the integrated terms above are of lower order in $s$. 
		Therefore, applying \eqref{propIntRMEp} to estimate the integrated terms in \eqref{propEst11},
		and recalling properties \eqref{E1}--\eqref{E2} for the function class $\cE$ (see \secref{secE}),
		we find that there exists function $\tilde f\in\cE$ s.th. for all $s\ge t>0$,
		\begin{align}\label{propEstp1}
			\big\lan N_{f,ts} ^p\big\ran_t \le\br{N_{f,0s}^p}_0+ Cs^{-1}\sum _{q=1}^p\big\langle N_{\tilde f,0s}^q \big\rangle_0 +  C s^{-n}  \TRem_p(s).  
		\end{align}

		Next, we simplify the r.h.s. of \eqref{propEstp1} and remove the $f$-dependence. 
		To begin with, analogous to \eqref{local-est4} and \eqref{locEstT}, we have,  for $s=(R-r)/c$ and any $p\ge1$, $f\in\cE$,
		\begin{align}
			  N_{f,0s}^p  &\le  N_{B_R}^p,\label{geoPest0}\\
			N_{B_r}^p  &\le  N_{f,ts}^p.\label{geoPestt}
		\end{align}
		Due to the algebraic identity $\sum_{q=1}^p (1+h)^{q-1}\le C_p(1+h^{p-1})$, we have
		\begin{align}\label{536'}
			\sum_{q=1}^p (N_{f,ts}+1)^{q-1}\le& C_p (1+N_{f,ts}^{p-1}).
		\end{align}
		In particular, \eqref{536'} implies
		\begin{align}
			\sum _{q=1}^p  N_{ f,0s}^q  \le 	\sum _{q=1}^p  (N_{ f,0s}^{q-1}+1)N_{ f,0s}\le C_p\del{ N_{ f,0s}+N_{ f,0s}^p}.
		\end{align}
		This, together with relation \eqref{geoPest0}, yields
		\begin{align}\label{551}
			\sum _{q=1}^p N_{ f,0s}^q\le C_p\del{{N_{B_R} }+{N_{B_R}^p }}.
		\end{align}
	{By relation \eqref{Nrel}, estimate \eqref{551} becomes
			\begin{align}\label{551'}
		\sum _{q=1}^p  N_{ f,0s}^q  \le C_p{{N_{B_R}^p }}.
	\end{align}
	Applying \eqref{551'} with $f$ replaced by $\tilde f$ and evaluating at the initial state yields the desired estimated for the second term in the r.h.s. of \eqref{propEstp1}.
}		

		Finally, we deal with the remainder term. Since $f$ is monotonously increasing (see \eqref{classE}), we have by definition \eqref{ftsDef} that $f_{ts}\le f_{0s}$ for all $v',\,t>0$. This, together with estimate \eqref{3122'}, implies 
		\begin{align}\label{ftsUB}
			f_{ts}(x)\le\chi_{B_R}(x),
		\end{align}
		and so, for any $p\ge0$,
		\begin{align}
			\label{geot'}
			{N_{f,ts}^p}\le {N_{B_{R}}^p}.
		\end{align}
		Recalling definition \eqref{TRemDef}, we conclude from \eqref{536'}, \eqref{geot'}, and \eqref{Nrel} that
		{\begin{align}
				\label{TRemEst}
			\TRem_p(s)=&\sup_{t\le s}{\sum_{q=1}^p\del{\langle (1+N_{f,ts})^{q-1}N_{B_{R-\eps s/2}}\rangle_{t}+\sum_{x\in B_{R-\eps s/2}}\sum_{y\in\Lambda}\abs{J_{xy}}|x-y|^{n+1}\langle (1+N_{f,ts})^{q-1}n_y\rangle_t }}\notag\\
				\le& C_p\sup_{t\le s}{ \del{\langle   (1+N_{f,ts}^{p-1})N_{B_{R-\eps s/2}}\rangle_{t}+\sum_{x\in B_{R-\eps s/2}}\sum_{y\in\Lambda}\abs{J_{xy}}|x-y|^{n+1}\langle (1+N_{f,ts}^{p-1})n_y\rangle_t }}\notag\\
\le& C_p\sup_{t\le s}\del{\langle (1+N_{B_R}^{p-1})N_{B_{R}}\rangle_{t}+\sum_{x\in B_{R}}\sum_{y\in\Lambda}\abs{J_{xy}}|x-y|^{n+1}\langle (1+N_{B_R}^{p-1})n_y\rangle_t }\notag\\
\le& C_p\sup_{t\le s}\del{\langle N_{B_R}^{p}\rangle_{t}+\sum_{x\in B_{R}}\sum_{y\in\Lambda}\abs{J_{xy}}|x-y|^{n+1}\langle (1+N_{B_R}^{p-1})n_y\rangle_t }
.
		\end{align}}
		This bounds the remainder term in \eqref{propEstp1}.
		
		Inserting \eqref{geoPestt}, \eqref{551}, and \eqref{TRemEst} into \eqref{propEstp1}, 
		we conclude the desired estimate \eqref{propEstp}.


		
		\qed
		
		\subsection*{Acknowledgment}	
		The authors thank J.~Faupin and I.~M.~Sigal for enjoyable and fruitful collaborations. The research of M.L.\ and C.R.~ is supported by the Deutsche Forschungsgemeinschaft (DFG, German Research Foundation) through grant TRR 352--470903074. 
		J.Z.~is supported by the National Key R \& D Program of China 2022YFA100740.

		\bibliographystyle{alpha}
		\bibliography{bibfile}
		
	\end{document}